\documentclass[letterpaper,11pt]{article}

\usepackage[margin=1in]{geometry}
\usepackage[utf8]{inputenc}
\usepackage[T1]{fontenc}
\usepackage{lmodern}
\usepackage[tbtags]{mathtools}
\allowdisplaybreaks
\usepackage{amssymb,amsthm}
\usepackage{dsfont}
\usepackage{cases}
\usepackage{thmtools}
\usepackage{thm-restate}
\usepackage{hyperref}
\usepackage{url}
\usepackage[svgnames]{xcolor}
\usepackage[capitalise,nameinlink]{cleveref}
\hypersetup{colorlinks={true},linkcolor={DarkBlue},citecolor=[named]{DarkGreen}}
\usepackage[font=footnotesize]{caption}
\usepackage[numbers,sort&compress]{natbib}
\usepackage{doi}
\usepackage{subcaption}
\usepackage{tikz}

\usepackage[utf8]{inputenc}
\usepackage[T1]{fontenc}
\usepackage{authblk}

\usepackage{extarrows}

\newtheorem{theorem}{Theorem}

\newtheorem{lemma}{Lemma}
\newtheorem{proposition}{Proposition}

\newtheorem{claim}{Claim}
\newenvironment{sketch}{%
	\proof}{\endproof}

\theoremstyle{definition}
\newtheorem{Definition}{Definition}

\usepackage[colorinlistoftodos,prependcaption]{todonotes}
\usepackage{relsize}
\usepackage{comment, url}
\usepackage{subcaption}

\newcommand{\eqstr}{\overset{\mathrm{N_0 \rightarrow \infty}}{=\joinrel=}}

\title{The Secretary Problem with Independent Sampling}

\author[1]{Jos\'e Correa}
\author[1]{Andr\'es Cristi}
\author[1]{Laurent Feuilloley}
\author[2]{Tim Oosterwijk}
\author[3]{Alexandros Tsigonias-Dimitriadis}
\affil[1]{Department of Industrial Engineering, Universidad de Chile\\ \url{correa@uchile.cl,andres.cristi@ing.uchile.cl,feuilloley@dii.uchile.cl}}
\affil[2]{School of Business and Economics, Maastricht University\\ \url{t.oosterwijk@maastrichtuniversity.nl }}
\affil[3]{Operations Research, TU Munich\\ \url{alexandros.tsigonias@tum.de}}
\date{}

\begin{document}

\maketitle
\setcounter{page}{0}
\thispagestyle{empty}

\begin{abstract} 
In the secretary problem we are faced with an online sequence of elements with values. Upon seeing an element we have to make an irrevocable take-it-or-leave-it decision. The goal is to maximize the probability of picking the element of maximum value. The most classic version of the problem is that in which the elements arrive in random order and their values are arbitrary. Here, the optimal algorithm picks the maximum value with probability at least $1/e$. However, by varying the available information, new interesting problems arise. For instance, in the full information variant of the secretary problem the values are i.i.d. samples from a known distribution. Naturally, the best possible success probability increases and turns out to be approximately $0.58$. Also, the case in which the arrival order is adversarial instead of random leads to interesting variants that have been considered in the literature.

In this paper we study both the random order and adversarial order secretary problems with an additional twist. The values are arbitrary, but before starting the online sequence we independently sample each element with a fixed probability $p$. The sampled elements become our information or history set and the game is played over the remaining elements. We call these problems the random order secretary problem with $p$-sampling (ROS$p$ for short) and the adversarial order secretary problem with $p$-sampling (AOS$p$ for short). Our main result is to obtain best possible  algorithms for both problems and all values of $p$. As $p$ grows to 1 the obtained guarantees converge to the optimal guarantees in the full information case. In the adversarial order setting, the best possible algorithm turns out to be a simple fixed threshold algorithm in which the optimal threshold is a function of $p$ only. Therefore, even knowledge of the total number of elements is unnecessary. Proving that this algorithm is optimal involves a novel technique, which boils down to analyzing a related game in a conflict graph over binary sequences. In the random order setting we prove that the best possible algorithm is characterized by a fixed sequence of time thresholds, dictating at which point in time we should start accepting a value that is both a maximum of the online sequence and has a given ranking within the sampled elements. Surprisingly, this sequence of time thresholds arises from a separable and convex optimization problem whose solution is independent of $p$.
\end{abstract}

\newpage


\section{Introduction}

In the secretary problem, probably the most well studied optimal stopping problem in computer science and discrete mathematics, we are faced with a randomly permuted sequence of $n$ elements with arbitrary values. The elements' values are revealed one at a time and upon receiving an element we need to make an irrevocable 
decision of whether we keep the value and stop the sequence or drop the value forever and continue observing the next. The goal is to maximize the probability of stopping with the largest value. For this problem the best possible success guarantee has long been known to be $1/e$ and the optimal algorithm is particularly simple: First we look at the first $n/e$ values without taking any of them, and then stop with the first value that is larger than all values seen so far \cite{D63,L61,F89}. In the last decades, the secretary problem, its variants and related optimal stopping problems such as the prophet inequality and the Pandora's box problem, have been considered fundamental building blocks of online selection problems in computer science and economics  \cite{KS77,KS78,W79,D18,BK19}. 

A particularly interesting question related to the classic secretary problem is how the success probability changes as the available information about the values and/or the underlying order of exploration does. Along these lines, already in the sixties Gilbert and Mosteller \cite{GM66} considered the so-called full information secretary problem in which we additionally know that the elements' values are i.i.d.\ random variables from a known distribution. For this variant, they showed how to compute the optimal stopping rule by dynamic programming and were able to conclude, numerically, that the best possible success probability is $\gamma\approx 0.5801$. In subsequent work, Samuels \cite{S91} finds an explicit expression for this quantity. Another interesting direction in the study of the secretary problem is to lift the random order assumption and consider a fixed (adversarial) order situation. An initial observation is that if the elements have arbitrary values then no algorithm can guarantee a larger than zero success probability. However, if we move towards the full information case, and assume that the values are not arbitrary but are independent realizations of random variables from known (different) distributions\footnote{If the distributions were not different then ordering would become equivalent to random order.}, then Allart and Islas \cite{AI15} showed that the optimal stopping rule can guarantee a success probability of $1/e$. The same model was recently considered by Esfandiari et al.\ \cite{EHLM20} who prove that if the values come from arbitrary independent distributions and the arrival order is random, then one can guarantee a success probability of $0.517$. Moreover, they show that if an additional distributional assumption is imposed, the success probability improves to $\gamma$.

In this paper we take a data-driven approach to the secretary problem, where the available information is parameterized by a sampling probability $p$. This allows to understand the tradeoffs between the amount of information available and the success probabilities that can be derived. In this context, data-driven versions of optimal stopping problems have been recently studied since the pioneering work of Azar et al.\ \cite{AKW14}. A notable result is that for the classic prophet inequality,\footnote{The classic prophet inequality asserts that when faced with a sequence of $n$ independent random variables, $X_1,\ldots,X_n$, a decision maker who knows their distributions and is allowed to stop the sequence at any time,  can obtain, in expectation, at least half the reward of a prophet who knows the values of each realization.} a single sample from each distribution is enough to achieve the optimal guarantee \cite{RWW20}. Also for the prophet secretary problem, the variant of the prophet inequality when the elements come in random order, one sample has been proved to be quite effective \cite{CCES20}. Recently, Kaplan et al.\ \cite{KNR20} study a model that is closest to ours. In their model there are $n$ arbitrary values and we sample a fraction $p$ of them at random. Then the values which were not sampled are presented to the decision maker in either random order or adversarial order. Kaplan et al.\ \cite{KNR20} design algorithms for maximizing the expectation, rather than the probability of picking the maximum, that translate into algorithms for data-driven versions of prophet inequalities. 

\paragraph{The problem.}
We consider a sampling model inspired by that of Kaplan et al.\ \cite{KNR20}. However, in our model, the sampling of each element is independent. Of course, for large $n$ the models are essentially equivalent. However, our independent sampling has two crucial advantages. On the one hand, independence makes many mathematical calculations a lot simpler and thus allows to obtain simpler expressions, while on the other hand, it allows to deal with instances of unknown size which is often the case in practical applications. In particular, several of our results hold if we do not know~$n$. A slight disadvantage of the independent sampling model is that it may happen that we end up sampling all $n$ elements. For consistency in this case we assume, by vacuity, that we {\em win} (i.e., pick the maximum), although this is not very restrictive since, as we will see, the difficult instances involve large values of $n$ for a fixed value of $p$.

More precisely, the problem we consider is described as follows. We are given $n$ elements with values $\alpha_1,\ldots,\alpha_n$, which are unknown to us, and an order $\sigma:[n]\to[n]$. Each element is sampled independently with probability $p$. Let $S$ be the (random) set of sampled elements and $V$ be the remaining elements, also referred to as the online set/elements. The elements in~$V$ are then presented to us in the order dictated by~$\sigma$. Once an element is revealed we either pick it and stop the sequence or drop it forever and continue. The goal is to maximize the probability of picking the maximum valued element in $V$. In the {\em adversarial order secretary problem with $p$-sampling} (AOS$p$) the order $\sigma$ is chosen by an adversary that knows all values $\alpha_1,\ldots,\alpha_n$ and the random sets $S$ and $V$.\footnote{Our results, and in particular the upper bounds on the success probability, remain true if the adversary knows all values $\alpha_1,\ldots,\alpha_n$ but not the result of the sampling process, i.e., she does not know the random sets $S$ and $V$.}
In the {\em random order secretary problem with $p$-sampling} (ROS$p$) the order $\sigma$ is just a uniform random permutation.

Given $n$ and an algorithm we define its {\em success probability} as the infimum over all values $\alpha_1,\ldots,\alpha_n$ of the probability that the algorithm stops with the maximum $\alpha_i \in V$. Moreover, the {\em success guarantee} of an algorithm is the infimum over all values of $n$ of its success probability. 

All algorithms considered in this paper are \emph{ordinal}, i.e., algorithms whose decision to stop at a given point depend only on the relative rankings of the values seen so far, and not on the actual values that have been observed, plus, possibly, on some external randomness. We observe that this is without loss of generality as for AOS$p$ and ROS$p$ general algorithms cannot perform better than ordinal algorithms. Indeed, as noted by Kaplan et al. \cite[Theorem 2.3]{KNR20} a result of Moran et al.\ \cite{MSM85} implies the existence of an infinite subset of the natural numbers where general algorithms behave like ordinal algorithms (for single selection ordinal objective functions such as ours). Therefore, and because the worst case performance of our algorithms is attained as $n\to \infty$, our bounds apply to general algorithms.

\paragraph{Our results.} 
For AOS$p$ we consider the following very simple algorithm. Upon observing the sample set~$S$ we take as threshold the value of its $k$-th largest element for $k=\left\lfloor \frac{1}{1-p} \right\rfloor$. Then we stop with the first element in $V$ whose value surpasses the threshold. If there are less than $k$ samples, the algorithm accepts the first online value (we define the $k$-th largest element from a set of less than $k$ elements as $-\infty$). We show that this algorithm achieves a success guarantee of $\left\lfloor \frac{1}{1-p} \right\rfloor p^{\left\lfloor \frac{1}{1-p} \right\rfloor}(1-p)$, so for instance for $p=1/2$ the guarantee evaluates to 1/4. Although the proof of this fact is relatively easy, what is more surprising is that this guarantee is best possible. To prove the latter we analyze a related optimal stopping problem, which we call the \emph{last zero problem}. Suppose an adversary picks a number of identical blank cards $n$. Then independently with probability $p$ each card is marked and you are informed about the total number of marked cards, but you ignore their position in the deck. Finally, one by one, you get to see the cards and whether they are marked or not. When you stop the sequence, you win if the card was the last blank card, otherwise you lose. By using a related conflict graph over possible sequences, we show that for this problem no ordinal algorithm can guess the last blank card with probability better than $\left\lfloor \frac{1}{1-p} \right\rfloor p^{\left\lfloor \frac{1}{1-p} \right\rfloor}(1-p)$. Then, we relate this problem to a different one, in which 
the objective is to guess the last number of an increasing sequence of unknown length. Finally, we go back to the original AOS$p$ by considering an adversary that picks a growing sequence which at some point in time decreases to a low value, and this time is difficult to guess.

It is worth noting that this simple best possible algorithm does not use knowledge of $n$ and, as opposed to most variants of the secretary problem, for AOS$p$ knowledge of $n$ is irrelevant in worst case terms. Moreover, we discuss the case in which $n$ is known but $p$ is unknown. Here it is quite natural that the algorithm works again by simply estimating $p$ using the size of the sample set. However, if neither $n$ nor $p$ are known, then no nontrivial success guarantee can be obtained.  

For ROS$p$ we obtain a randomized algorithm with best possible success guarantee that works as follows. First, we assign to each of the $n$ elements a uniformly random arrival time in the interval $[0,1]$, which implies that the elements arrive in uniform random order. All elements whose arrival time is less than $p$ are placed in the sample set $S$. Then we find a sequence of time thresholds $0<t_1<t_2<\cdots<1$, dictating that if an element's arrival time is between $t_i$ and $t_{i+1}$, we stop if its value is the maximum among elements arriving after $p$ and it is among the $i$ largest values of all elements seen so far. To obtain the success guarantee of this algorithm we first prove that for a fixed sequence $0<t_1<t_2<\cdots<1$, the success guarantee of the algorithm decreases with $n$. Then we write the optimization problem over the time thresholds, and interestingly, this turns out to be a separable concave optimization problem with a very simple solution. Moreover, the solution is universal in the sense that it does not depend on $p$. The resulting guarantee is thus easily computed and grows from $1/e$ when $p=0$ to $\gamma\approx 0.58$ as $p\to 1$. We also prove that this is a best possible algorithm. To this end we first argue that ordinal algorithms in our model are essentially equivalent to a ranking function that determines what {\em global} ranking an element, which is a {\em local} maximum, should have in order to accept it. Here, by global ranking we mean the ranking an element has among all samples and values revealed so far, and local ranking refers only to the values revealed and not to the samples. Finally, as $n$ grows, this ranking function converges to a sequence of time thresholds as we defined them. 

\cref{fig:k_max_algo} illustrates the success guarantee for our problems. For AOS$p$ it can be observed that the success guarantee can be bounded below by the function $p^{(1/(1-p))}$ and bounded above by $\frac{p-1}{\log p} \cdot p^{-1/\log p}$. 

\begin{figure}[t]
	\begin{center}
		\includegraphics[scale=0.65]{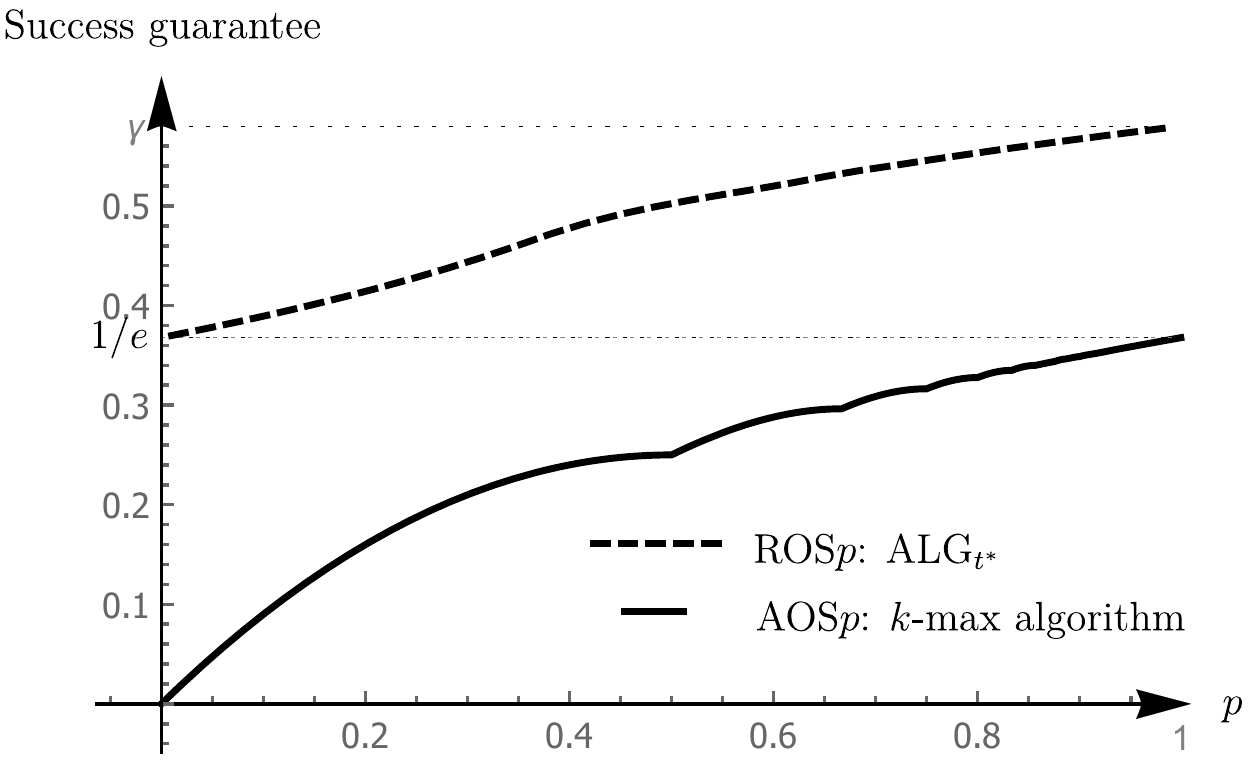}
		\caption{\label{fig:k_max_algo} The best possible success guarantee for ROS$p$ and AOS$p$ as a function of $p$.}
	\end{center}
\end{figure}

\paragraph{Further related literature.} 
An interesting connection arises between our model and results when $p$ is close to 1, and the so-called full information case. First, recall that Gilbert and Mosteller \cite{GM66} obtained the optimal algorithm with worst case performance $\gamma$ (see also \cite{S82,S91}), in the secretary problem where the elements' values are taken as i.i.d.\ random variables from a known distribution. It may thus seem natural that our guarantee matches this quantity as $p\to 1$. However, this is far from obvious. Indeed, for the prophet inequality with i.i.d.\ values from an unknown distribution (a model that arguable gives more information than ours) Correa et al.\ \cite{CDFK19} proved that with $O(n^2)$ samples one can achieve the best possible performance guarantee of the case with known distribution, and only very recently Rubinstein et al.\ \cite{RWW20} improved this to $O(n)$ samples. This is in line with our result here since for $p$ close to, but strictly less than 1, the size of the sample set is linear in the size of $V$.

A more intriguing connection to the full information case pops up in the adversarial order case. In this context, Allart and Islas \cite{AI15}, and independently Esfandiari et al.\ \cite{EHLM20}, considered the adversarial order secretary problem in which an adversary chooses $n$ distributions $F_1,\ldots,F_n$. Then, independent values are drawn from these distributions and sequentially uncovered. A decision maker who knows $F_1,\ldots,F_n$ needs to stop at the maximum realization. They prove that the optimal stopping rule is a simple single threshold algorithm and the best possible success guarantee equals $1/e$. Although this problem has a similar flavor as our AOS$p$, and the optimal guarantee is the same, we are unaware of a precise connection. 

On the other hand, our last zero problem, used as a tool for AOS$p$, is related to an old optimal stopping problem first studied by Bruss \cite{B00}.
We face a sequence of $n$ independent Bernoulli random variables where we know $n$ and the distributions, and we want to stop with the last zero. Bruss obtains the optimal stopping rule for this problem, which also turns out to be a simple threshold rule. Our last zero problem is simpler in that the Bernoulli random variables are homogeneous. However, rather than knowledge of $n$ we only know the total number of ones. This subtle difference makes the problem substantially different.

Another very recent line of work \cite{KM20,BGSZ20} studies robust or semi-random versions of the classical secretary problem. The main idea is that the problem input should be a mix of stochastic and adversarial parts. More specifically, in their (similar) models some of the elements arrive at adversarially chosen times and the rest at times uniformly randomly drawn from $[0,1]$. Their objective functions (and in some cases also the benchmarks) are quite different from ours. The authors in \cite{KM20} consider the knapsack secretary problem in this mixed model, while \cite{BGSZ20} design algorithms for selecting $k$ items or maximizing the expectation under various matroid or knapsack constraints. It would probably be interesting to incorporate their ideas in our setting and study a problem which interpolates between ROS$p$ and AOS$p$.

\paragraph{Outline of the paper.}
\cref{sec:OverviewAOS} presents an overview of the techniques and results for the adversarial order case, while \cref{sec:RO} does the same for the random order case. Then, \cref{sec:Parameters-overview} presents some insights into the results that can be obtained if we assume different knowledge of the parameters. Like the previous sections, this section contains proof sketches. The full proofs of these three sections can be found in \cref{sec:proof-AO}, \cref{sec:proof-RO} and \cref{sec:proof-knowledge} respectively.

\section{Overview for adversarial order}\label{sec:OverviewAOS}

This section introduces the main techniques and ideas behind Theorem~\ref{thmAO}. All details can be found in \cref{sec:proof-AO}.
Recall that we defined the \emph{$k$-max algorithm} as follows: the $k$-th largest value of the sampled elements is set as a threshold, and the algorithm accepts the first element in the set $V$ of online values whose value surpasses this threshold. If there are less than $k$ sampled elements, then the algorithm accepts the first online element\footnote{Recall that we define the $k$-th largest element from a set of less than $k$ elements as $-\infty$.}. From now on, we take the $k$-max algorithm with $k=\left\lfloor \frac{1}{1-p} \right\rfloor$.

\begin{restatable}{theorem}{thmAO}
	\label{thmAO}
	Let $k=\left\lfloor \frac{1}{1-p} \right\rfloor$.  Then the $k$-max algorithm achieves a guarantee of $k p^k(1-p)$ for AOS$p$. Furthermore, no algorithm can achieve a better success guarantee.
\end{restatable}

When $p$ tends to 0, the guarantee naturally tends to zero: If there are very few samples, the problem becomes the secretary problem with adversarial order, where basically nothing can be done.
When $p$ is close to $1$, the success guarantee approaches $1/e$, which is the performance obtained for the secretary problem when one knows the distribution of the values of the elements~\cite{AI15,EHLM20}. See \cref{fig:k_max_algo}.

The proof of the guarantee of the algorithm is easy and appears in this overview section (see \cref{lem:k-max-guarantee}). The proof of its optimality is more advanced and requires new tools. We give a detailed overview of the proof in \cref{subsec:overview-AO-negative} and the entire proof is in \cref{sec:proof-AO}. 
A surprising fact of this proof is the following: when proving the negative result, it is enough to focus on the special case where the values of the elements are increasing (thus where the player aims to get the last element), with the twist that the player does not know the total size $n$ of the instance.




\subsection{The success guarantee of the $k$-max algorithm}

Being a simple threshold algorithm, the main question to answer is what value of $k$ is appropriate.
Intuitively, the bigger the value of $p$, the higher the probability that the largest valued elements are sampled. Therefore, we should lower the threshold as $p$ grows. As is the case for many threshold algorithms, there is a trade-off between (1) setting the threshold too low and risking acceptance of an element that does not have the maximum online value, and (2) setting it too high and risking finishing the game without selecting any element. The following lemma establishes the performance of the algorithm for the value $k=\left\lfloor \frac{1}{1-p} \right\rfloor$. 

\begin{lemma}\label{lem:k-max-guarantee}
	For a given $p$, the $k$-max algorithm chooses the element of the online set with maximum value with probability 
	$\left\lfloor \frac{1}{1-p} \right\rfloor p^{\left\lfloor \frac{1}{1-p} \right\rfloor}(1-p)$.\footnote{Observe that this lemma still holds in the setting where the order of the online elements is determined by the adversary after sampling, since our algorithm is order oblivious.}
\end{lemma}

\begin{proof}
	Note that the $k$-max algorithm wins in an instance if \emph{exactly one} of the $k$ largest values of the adversarial input ends up in the online set and the $(k+1)$-th largest ends up in the sample set. Thus, an instance in which the algorithm is successful is exactly a sequence ending in $k$ sampled elements plus one online element that is somewhere in the last $k$ entries of the sequence.
	The probability that this happens equals $kp^k(1-p)$. The lemma follows by substituting the value $k=\left\lfloor \frac{1}{1-p} \right\rfloor$.
\end{proof}

\subsection{The negative result}
\label{subsec:overview-AO-negative}

We now sketch the proof for the negative result of Theorem~\ref{thmAO}, which consists of several steps. We start by considering the special case where \emph{the algorithm does not know $n$}. Let us make precise what we mean by this. Consider an algorithm~$A$ and two instances $I_1$ and $I_2$ of different sizes $n_1$ and~$n_2$ respectively, but with the same value of $p$. Suppose that $A$ happens to face the exact same set of samples in both instances, and is currently facing an online element of the same value in both instances. Thus, up to this point, $A$ has access to exactly the same information (and possible beliefs over the size of the instance). Therefore, $A$ needs to make the exact same (possible randomized) decision in both situations, independent of $n_1$ or $n_2$.

For our main steps, we start by showing that we can focus on a simpler problem that we call the last zero problem. For this problem, we prove the negative result with some additional assumptions. We then remove the assumptions one by one, each time generalizing the proof one step further, until we get the proof of Theorem~\ref{thmAO} for the case where $n$ is unknown. Finally, in the second phase, we show that allowing the algorithm to know $n$ basically does not help (in worst case terms). The proof can be found in \cref{sec:proof-AO}.

The \emph{last zero problem} with probability $p$ is the following: First, an adversary chooses an integer $n$, and $n$ blank cards are created and form a deck. Second, a referee takes the deck, and writes a number on each card: 1 with probability $p$ and $0$ with probability $1-p$. The referee gives the number of 1s to the player.
Finally, the cards are presented to the player one after the other in the order of the deck. Upon every card, the player needs to decide whether to stop the sequence or not, and she wins if she stops at the last 0. She does not know the value of $n$, but she does know $p$ and the number of 1s. An instance for the last zero problem can be represented by a sequence of bits.

The following proposition highlights the connection between the last zero problem and AOS$p$. With the increasing case of AOS$p$ we mean the special case of the problem AOS$p$ where the elements are presented to the algorithm in increasing order of their values.

%

\begin{restatable}{proposition}{proplastzeroequivalent}
	\label{prop:last-zero-equivalent}
	The last zero problem and the increasing case of AOS$p$ are equivalent. Therefore, any negative result for the last zero problem also holds for AOS$p$.
\end{restatable}

\begin{proof}
We show that an algorithm for picking the element with maximum value in the increasing case of AOS$p$ has the same success probability in the last zero problem, and the other way round.

$(\Rightarrow)$ Assume that we know that in AOS$p$ the adversary is going to present the online set in increasing order. Therefore we need to fix an ordinal algorithm with the goal of picking the last element in the increasing sequence. 
Every time an element in $V$ is revealed, the algorithm knows how many online elements it saw in total and how many sampled elements have larger or smaller values compared to the value of this online element. Moreover, the value of $p$ creates some possible beliefs over the size of the instance. This knowledge guides the (possibly randomized) decision of the algorithm on whether to stop with the element just observed.

In the last zero problem, each revealed 0 of the binary sequence corresponds to an online element. Furthermore, since we are given the total number of 1s beforehand, we know how many 1s are before and after each revealed 0 in the sequence. This information corresponds to the relative ranking of an elements value in $V$ among the values of sampled elements. Finally, $p$ equals the probability that a 1 was written on a card, independently of the others.

An algorithm for AOS$p$ takes as input the relative ranking of the values $r_1 > r_2 > \dots > r_t$ in $S$ and $V$ seen so far at each time step $t$ and outputs a stopping rule $\tau$ which gives a certain success probability. In particular, since the algorithm is ordinal, it does not even need to see the actual values of the sampled elements; all it needs to know is the ranking of a revealed element among the sampled ones. If we apply the same algorithm to the last zero problem (with the input now being the total number of 0s and 1s seen so far and the total number of 1s), we get the same success probability of picking the last 0.

$(\Leftarrow)$ Consider an algorithm for maximizing the probability of picking the last zero in the last zero problem. At each time step $t$, an algorithm $ALG_\tau'$ takes as input the given probability $p$, the total number $k$ of 1s (also given) and how many 0s and 1s have been seen so far. Consider a stopping rule $\tau'$ that decides whether to stop at each revealed 0, and that attains a certain success probability.

In the increasing case of AOS$p$ each element in $S$ corresponds to a 1 and each online element to a 0. The total number of 1s represents the cardinality of the set $S$. Each time a 0 is observed (and we know its rank among the 1s), it translates to learning how many samples have smaller and how many have larger value than the online element just observed. Remember that since the sequence in AOS$p$ is increasing, we win if we stop with the last online element. We can now conclude that an algorithm for the last zero problem with a certain success probability can be used as an ordinal algorithm to solve the increasing case of AOS$p$ with the same success probability. 

Since the increasing case is a specific instance for AOS$p$, a negative result for the last zero problem implies a hardness result for AOS$p$.
\end{proof}

From now on we focus on the last zero problem. 
We start with proving an upper bound for the special case of deterministic algorithms for $p=1/2$. For this case of the last zero problem we introduce the \emph{no-zero rule} that specifies that if there are no online elements (i.e., all $n$ elements are sampled), the player loses. This will be a useful rule for the sake of the proofs. As we will see, this decision actually becomes irrelevant for the generalization of the proof. Therefore, it poses no problem that this contradicts the assumption made for AOS$p$ where we win in such an instance.

The following proposition holds under the no-zero rule and starting from $n=1$.
\begin{restatable}{proposition}{prop:warm-up}
	\label{prop:warm-up}
	For the last zero problem with $p=1/2$, no deterministic algorithm can achieve a better success guarantee than $k$-max (with the no-zero rule) for AOS$p$.
\end{restatable}

This proposition and its proof sketch are presented here to introduce informally the tools we will use. Its statement can be generalized to consider instances of size larger than some chosen $N_0$ (cf. \cref{prop:p-half}). 
\cref{sec:proof-AO} will prove this generalization directly.

\begin{sketch}
	For $p=1/2$, the $k$-max algorithm gets a guarantee of 1/4. Suppose that there is an algorithm that achieves a guarantee strictly better than 1/4.
	As a start, consider the decision of the algorithm when the adversary chooses $n=1$.
	Then, there are two instances (after sampling) which both occur with probability 1/2.
	The first possibility is that the instance is 0.
	Then the player knows that there is no 1 in the instance, and is first presented a 0. The second possibility is that the instance is 1. Then
	the player knows there is a 1 in the instance, and is announced from the start that the game is finished.
	
	In the second case, the player loses because of the no-zero rule. Thus, to achieve at least 1/4 for every $n$, the player needs to win in the first case. This means that when the player is presented with not a single 1, and sees a first 0, she stops.
	
	Here comes the key observation. Suppose that the adversary chose $n=2$ and the sampling resulted in the instance 00.
	Now again the player is presented with not a single 1, and again sees a first 0. From the above, we already deduced that she needs to stop at this first~0.
	Indeed, from the point of view of the player, this is exactly the same situation as in the case where the instance was 0, because the player does not know $n$. In other words, these two situations are \emph{indistinguishable} from the perspective of the player, and she has to make the same decision. In the case of 00, this decision is wrong as the last 0 is the second 0, hence the player loses. We call such a situation a \emph{conflict} between the instances 0 and 00. 
	
	Note that conflict works in both directions. If the player had a strategy that would make her win in 00, then after the first 0, she would wait, which would make her lose in the instance 0.
	
	Let us give yet another example of conflict, for the instances 01 and 001. On instance 001 the player receives a first 0, and knows that there is one 1. This is exactly the same information as in the instance 01 when it starts. If she stops on this element then she wins in~01 but loses in 001. On the other hand, if she waits and then stops on the next 0, she loses in 01 but wins in 001. Moreover, if she continues to wait she loses in both instances.
	
	More generally, for every pair of instances there is a fairly simple criterion in each of the two directions to see if they are in conflict or not. 
	In particular, it is enough to decide the conflict between instances whose sizes differ only by 1. 
	Indeed, two instances $s$ and $s'$ of sizes $n$ and $n+q$ respectively are in conflict if and only if there is a series of conflicts $(s,s_1)$, $(s_1,s_2)$, ..., $(s_{q-1},s')$, where $s_i$ has size $n+i$ (cf. \cref{lem:monotone-path-conflict}).
	Then we can define the (infinite) \emph{conflict graph} whose nodes are all possible instances and the edges represent the conflict between nodes of adjacent sizes.
	The conflict graph for size $n=1$ to $n=4$ is represented in \cref{fig:table1234}. 
	On this graph, we can represent an algorithm as a choice of instances in which it wins. Such \emph{selected instances} cannot be in conflict. In other words, they cannot be linked by a monotone path, where monotone means that the path goes from left to right without changing direction.  
	
	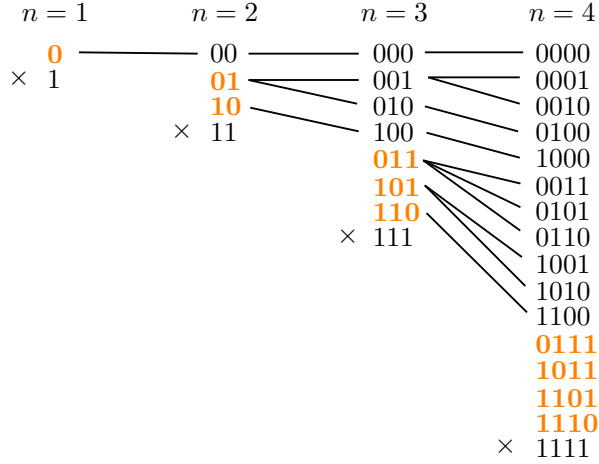
\begin{figure}[t]
		\begin{center}
			\scalebox{0.95}{
			\tikzset{every picture/.style={line width=0.75pt}} 

\begin{tikzpicture}[x=0.75pt,y=0.75pt,yscale=-0.7,xscale=0.9]

\draw    (150.38,75.07) -- (220.38,76.07) ;
Straight Lines [id:da4619298462748319] 
\draw    (250.38,76.07) -- (315.38,76.07) ;
Straight Lines [id:da05969078291676655] 
\draw    (354.38,75.07) -- (412.38,75.07) ;
Straight Lines [id:da5722025373907577] 
\draw    (250.38,96.07) -- (315.38,96.07) ;
Straight Lines [id:da3964781263119016] 
\draw    (250.38,96.07) -- (314.38,114.07) ;
Straight Lines [id:da41751410254202237] 
\draw    (251.38,117.07) -- (315.38,135.07) ;
Straight Lines [id:da6377289947575543] 
\draw    (356.38,94.07) -- (410.38,94.07) ;
Straight Lines [id:da06237270783813331] 
\draw    (356.38,94.07) -- (410.38,114.07) ;
Straight Lines [id:da3100942078625778] 
\draw    (354.38,116.07) -- (409.38,135.07) ;
Straight Lines [id:da921883519121422] 
\draw    (355.38,136.07) -- (410.38,155.07) ;
Straight Lines [id:da6279254973010989] 
\draw    (353.38,157.07) -- (410.78,192.65) ;
Straight Lines [id:da8661561261267355] 
\draw    (353.38,157.07) -- (410.78,174.65) ;
Straight Lines [id:da3143073901356962] 
\draw    (353.38,157.07) -- (410.65,210.4) ;
Straight Lines [id:da6829831790255821] 
\draw    (354.38,176.07) -- (411.65,230.45) ;
Straight Lines [id:da4591739894442768] 
\draw    (354.38,176.07) -- (413.65,252.45) ;
Straight Lines [id:da38155270385134765] 
\draw    (355.38,197.07) -- (414.65,272.45) ;

\draw (115,36) node [anchor=north west][inner sep=0.75pt]   [align=left] {$n=1$};
 Text Node
\draw (215,36) node [anchor=north west][inner sep=0.75pt]   [align=left] {$n=2$};
 Text Node
\draw (315,36) node [anchor=north west][inner sep=0.75pt]   [align=left] {$n=3$};
 Text Node
\draw (414,36) node [anchor=north west][inner sep=0.75pt]   [align=left] {$n=4$};
 Text Node
\draw (130,67) node [anchor=north west][inner sep=0.75pt]   [align=left] {\textbf{\textcolor{orange}{0}}};
 Text Node
\draw (130,87) node [anchor=north west][inner sep=0.75pt]   [align=left] {1};
\draw (226,67) node [anchor=north west][inner sep=0.75pt]   [align=left] {00};
\draw (226,88) node [anchor=north west][inner sep=0.75pt]   [align=left] {\textbf{\textcolor{orange}{01}}};
\draw (226,127) node [anchor=north west][inner sep=0.75pt]   [align=left] {11};
\draw (226,107) node [anchor=north west][inner sep=0.75pt]   [align=left] {\textbf{\textcolor{orange}{10}}};
\draw (322,67) node [anchor=north west][inner sep=0.75pt]   [align=left] {000};
\draw (322,88) node [anchor=north west][inner sep=0.75pt]   [align=left] {001};
\draw (322,127) node [anchor=north west][inner sep=0.75pt]   [align=left] {100};
\draw (322,107) node [anchor=north west][inner sep=0.75pt]   [align=left] {010};
\draw (322,147) node [anchor=north west][inner sep=0.75pt]   [align=left] {\textbf{\textcolor{orange}{011}}};
\draw (322,168) node [anchor=north west][inner sep=0.75pt]   [align=left] {\textbf{\textcolor{orange}{101}}};
\draw (322,207) node [anchor=north west][inner sep=0.75pt]   [align=left] {111\\};
\draw (322,187) node [anchor=north west][inner sep=0.75pt]   [align=left] {\textbf{\textcolor{orange}{110}}};
\draw (418,67) node [anchor=north west][inner sep=0.75pt]   [align=left] {0000};
\draw (418,88) node [anchor=north west][inner sep=0.75pt]   [align=left] {0001};
\draw (418,127) node [anchor=north west][inner sep=0.75pt]   [align=left] {0100};
\draw (418,107) node [anchor=north west][inner sep=0.75pt]   [align=left] {0010};
\draw (418,147) node [anchor=north west][inner sep=0.75pt]   [align=left] {1000};
\draw (418,168) node [anchor=north west][inner sep=0.75pt]   [align=left] {0011};
\draw (418,207) node [anchor=north west][inner sep=0.75pt]   [align=left] {0110};
\draw (418,187) node [anchor=north west][inner sep=0.75pt]   [align=left] {0101};
\draw (418,227) node [anchor=north west][inner sep=0.75pt]   [align=left] {1001};
\draw (418,248) node [anchor=north west][inner sep=0.75pt]   [align=left] {1010};
\draw (418,287) node [anchor=north west][inner sep=0.75pt]   [align=left] {\textbf{\textcolor{orange}{0111}}};
\draw (418,267) node [anchor=north west][inner sep=0.75pt]   [align=left] {1100};
\draw (418,307) node [anchor=north west][inner sep=0.75pt]   [align=left] {\textbf{\textcolor{orange}{1011}}};
\draw (418,328) node [anchor=north west][inner sep=0.75pt]   [align=left] {\textbf{\textcolor{orange}{1101}}};
\draw (418,367) node [anchor=north west][inner sep=0.75pt]   [align=left] {1111};
\draw (418,347) node [anchor=north west][inner sep=0.75pt]   [align=left] {\textbf{\textcolor{orange}{1110}}};
\draw (107,85.8) node [anchor=north west][inner sep=0.75pt]   [align=left] {$\displaystyle \times $};
\draw (203,125.8) node [anchor=north west][inner sep=0.75pt]   [align=left] {$\displaystyle \times $};
\draw (301,205.8) node [anchor=north west][inner sep=0.75pt]   [align=left] {$\displaystyle \times $};
\draw (394,364.8) node [anchor=north west][inner sep=0.75pt]   [align=left] {$\displaystyle \times $};

\end{tikzpicture}}
			\caption{An illustration of the first four layers of the conflict graph. \label{fig:table1234}}
		\end{center}
	\end{figure}
	
	For $n=1$, we denote by a cross the fact that the player will never win in the instance which consists of one 1, because of the no-zero rule. We write 0 in orange to denote that the player wins in this instance, as she decides to select the last (and only) 0.
	
	Let us now consider more systematically all the instances of size 2. They all have a probability of occurring of $1/4$. We already know that both 00 and~11 cannot be selected (because of the conflict to the left and the no-zero rule, respectively). Thus, to achieve strictly more than 1/4, the player needs to win in both 01 and 10. Consequently, these instances need to be selected in the conflict graph.
	
	Now, for $n=3$, the player loses in 000, 001, 010 and 100 because of conflicts, and on 111 because of the no-zero rule. Therefore, she \emph{must} win in 011, 101, and~110, since each instance has probability~$1/8$.
	Finally, for size 4, we can use the same kind of argument as before, to show that the player loses in all instances except 0111, 1011, 1101 and 1110. But these are only four cases out of sixteen and thus, the player cannot strictly beat the 1/4 bound if the adversary chose $n=4$. And this is a contradiction.
\end{sketch}

There are several limitations to this first proof:
\begin{enumerate}
	\setlength\itemsep{0pt} 
	\item \label{item:no-online} The no-zero rule is arbitrary and it should be removed.
	\item \label{item:start-zero} The fact that the proof is only considering small sizes is a weakness, in the sense that it does not take into account algorithms which could possibly have a much better success guarantee than k-max, when starting from some large $N_0$.
	\item \label{item:p-half} The sampling probability is fixed to 1/2 instead of taking any value in $(0,1).$
	\item \label{item:deterministic} The bound only applies to deterministic algorithms.
\end{enumerate}

We continue by addressing the two first problems.
We design a proof that also works by starting from an arbitrary $N_0$ and not necessarily from 1. This also solves the first problem, as it makes the probability of the case with no zeros negligible for large enough $N_0$ (for size $n$, this case is just one out of $2^n$).

\begin{restatable}{proposition}{propphalf}
	\label{prop:p-half}
	For the last zero problem with $p=1/2$, no deterministic algorithm can have a better success guarantee than $k$-max, even if we consider only instances of size larger than $N_0$, for any $N_0$.
\end{restatable}

We now present a proof sketch. The full version of the proof can be found in  \cref{subsec:p-half}.

\begin{sketch}
	Again, consider a strategy that has a strictly  better guarantee than 1/4, starting from some arbitrary $N_0$.
	In the proof sketch of \cref{prop:warm-up}, there was somehow no choice on which node to select: for size 1, we had to select 0, for size 2 we had to select 01 and 10 etc. This is not the case anymore if we start with some large size $N_0$, since for this size only $1^{N_0}$ cannot be selected due to the no-zero rule. So in principle, we could select all nodes of size $N_0$ except this one to achieve a very high performance for this size. But this would mean a lot of conflicts later, and would prevent a good performance for (many) larger sizes. Therefore, the design of a strategy is a trade-off between having a good performance for the size at hand (if we design the strategy from left to right) and ensuring that there are still instances without conflicts for larger sizes. 
	
	The principle of the proof is similar to the one in \cref{prop:warm-up}: we start from size $N_0$ and we assume that the player selects a set of instances whose probability is in total at least  $1/4+\varepsilon$. Then we move on to the next size, consider which instances are in conflict, and the player again selects a set of instances whose probability is at least $1/4+\varepsilon$. We will see that if we continue accordingly, then for some size 
	there are not enough instances that are not in conflict to allow for a performance of $1/4+\varepsilon$.
	
	A problem we face in the analysis is that there are many ways in which the player could select its instances, and that these many different ways would lead to very different conflicts. In particular, we can already note in \cref{fig:table1234} that different nodes have different degrees (where a node's degree is the number of nodes connected to its right). The degree of a node is crucial as a large degree implies a lot of conflicts.
	
	The key observation 
	is that for any size $n$, if an algorithm selects a node $v$ of degree $d$ but does not select a node $v'$ of degree $d'<d$ that is not in conflict with a node of smaller size, then the algorithm can deselect $v$ and select $v'$. This gives rise to a new strategy which is still valid (in the sense that it does not select two instances in conflict) and the performance of this new algorithm is the same.
	
	In other words, when designing a strategy, one can always pick the small degrees first. This leads to a canonical form for strategies. And for such strategies, arguments similar to the proof of \cref{prop:warm-up} yield the result.
\end{sketch}

Let us now move on to the case of general $p$, that is, overcome the third limitation. This case is a bit more complicated, because when $p \neq 1/2$ the instances of the same size do not have the same probability of occurring.
For example, for $p=3/4$, it is better for an algorithm to succeed in the instance $1^k0$ than to succeed in the instance $0^{k+1}$, as the first has probability $(3/4)^k(1/4)$ to occur and the second probability $(1/4)^{k+1}$.
From a technical perspective, this means that in the conflict graph the nodes now have weights. 
But the main change is that, as not all nodes have the same weight, the key observation above for $p=1/2$ does not work any more since there are nodes of small degree that have smaller weight than nodes of larger degree. Swapping nodes solely based on their degree therefore does not guarantee the same success guarantee.

To solve this problem, we use another albeit similar approach. 
We define a modification of a strategy that can only make it better, and this modification leads to another set of canonical strategies for which similar ideas as before can be applied to obtain the results.

\begin{restatable}{proposition}{prop:general-p}
	\label{prop:general-p}
	For the last zero problem with \emph{any} value of $p$, no algorithm (deterministic or randomized) can achieve a better success guarantee than the $k$-max algorithm, starting from an arbitrary $N_0$.
\end{restatable}

\begin{sketch}
	Our crucial observation here 
	is that there is a threshold $D$, depending only on $p$ such that: (1) if a node of degree larger than $D$ is selected, then the strategy can only be improved by deselecting it and selecting its children and (2) if a node of degree smaller than $D$ is not selected, then the strategy can only be improved by selecting this node and deselecting all its descendants.
	Note that these changes involve instances of multiple sizes.
	
	A consequence of this is that these moves can decrease the performance for some given size and improve the performance for some other size, which does not correspond to our performance measure (which is an infimum of the guarantees over all sizes). 
	We overcome this issue in the following way.
	If an algorithm can always perform better than $k$-max, then it also performs better on average. Our strategy modification always improves the average performance, and we prove that after applying these modifications we obtain a strategy whose average performance is not better than the average performance of $k$-max algorithm. This leads to a contradiction again. 
	
	Finally, randomized algorithms can be analyzed with the same kind of tools. The only difference is that strategies select nodes with some probability instead of selecting it either fully or not. All the statements can then be adapted to this ``fractional'' version in a pretty straightforward way.
\end{sketch}

Finally, we extend the results to the scenario in which the player knows the value of $n$. (Remember that an upper bound in our context is a hardness result.)

\begin{restatable}{proposition}{prop:knowing-n}
	\label{prop:knowing-n}
	An upper bound on the success guarantee for the last zero problem implies the same bound for AOS$p$ even in the general case \emph{where the player knows the total size}.
\end{restatable}

\begin{sketch}  
	To show this strengthening of \cref{prop:last-zero-equivalent}, we use a simple trick: instead of doing the reduction with instances of increasing values, we consider instances that have first an increasing part and then elements with very low values. 
	The adversary can choose at which point to switch from one regime to the other. This type of instances basically mimics the previous case, as the player wins if she picks the last online element of the first part.
	This makes it a bit more complicated, as the samples from the second part give some additional indication regarding the exact moment that the switch occurs, but it turns out that our previous proof is robust to this.

	More concretely, now that the player knows the value of $n$, we need a slight variant of the last zero problem, which we call the \emph{colored last zero problem}: First, an adversary picks two integers $m$ and $n$, with $m\leq n$. A sequence of bits of length $n$ is created where every entry independently has value 1 with probability~$p$ and 0 otherwise. We color the entries $1$ to $m$ with red, while the entries $m+1$ to $n$ are colored blue. The player is given the size $n$, the number of red 1s and the number of blue 1s. Then, the player is presented with the bits one after the other, and for each of them decides whether to stop or to continue. The player wins if she stops on the last red 0 of the sequence.
\end{sketch}

Note that for the colored last zero problem, $m$ basically plays the role that $n$ was playing before. This change leads to a modified conflict graph. Now, the different layers of the conflict graph correspond to the various sizes of $m$ in this case, and there is a separate conflict graph for each value of $n$. Note that the conflict graph has a finite number of layers as $m$ varies between $1$ and $n$. A node of the graph is a couple $(S,b)$, where $S$ is a sequence of bits of length $m$, that represents the sequence of red bits, and $b$ is an integer that represents the number of blue 1s. The exact positions of 0s and~1s in the blue bits are irrelevant, only the total number of blue 1s matters. Using these modified tools, we can follow very similar arguments as before and prove the hardness result for the case of known $n$ as well. 


\section{Overview for random order}
\label{sec:RO}

In this section we study the second problem of this paper: the random-order secretary problem with $p$-sampling, ROS$p$. 
To analyze this case it is useful to have the following equivalent point of view. We assign a uniformly random arrival time $\tau_i$ to each of the $n$ elements in the interval $[0,1]$. If $\tau_i<p$ we add $i$ to $S$ and otherwise we add it to $V$. Then the elements in $V$ are revealed in the order of the $\tau_i$'s.
Clearly, $\tau_i<p$ with probability $p$, so each value is in~$S$ independently with probability $p$. It is also clear that the resulting order is uniformly random. Therefore, any algorithm for the original formulation can be applied to this one. Conversely, an algorithm for this formulation can be transformed into a randomized algorithm for the original one, by sampling $|S|$ uniform arrival times in $[0,p]$ and $|V|$ uniform arrival times in $[p,1]$.

Consider the following family of algorithms. We fix a sequence $t=(t_i)_{i\in \mathbb{N}}$ such that $0\leq t_1<t_2<\cdots<1$. 
Between times $t_k$ and $t_{k+1}$ the algorithm $ALG_t$ sets as a threshold the $k$-th largest sampled value. More precisely, suppose the value $\alpha_i$ is revealed and assume $t_k\leq \tau_i<t_{k+1}$. $ALG_t$ accepts $\alpha_i$ if it is the largest among the values from $V$ seen so far, and is greater than the $k$-th largest value from $S$. For simplicity, if $|S|<k$ we define the $k$-th largest value of $S$ as $-\infty$. We prove that the best possible success guarantee is attained in this family.

\begin{theorem}\label{thmRO}
  There exists a universal sequence $t$, independent of $p$ and $n$, such that $ALG_t$ obtains the best possible success guarantee for ROS$p$. Furthermore, when $p=0$ this guarantee is equal to $1/e$, and when $p$ tends to 1, the guarantee tends to $\gamma\approx 0.58$, the optimal success guarantee in the full-information secretary problem.\footnote{The optimal guarantee $\gamma\approx 0.58$ was first obtained numerically by Gilbert and Mosteller~\cite{GM66}. An explicit formula for $\gamma$ was later found by Samuels~\cite{S82,S91}.} 
\end{theorem}

We prove this theorem in two main steps. First, we find the sequence $t^*$ that maximizes the success guarantee of $ALG_t$. Then, we find an expression for the optimal success probability when $p$ and $n$ are given, and prove that for fixed $p$ it converges to the success guarantee of $ALG_{t^*}$ when $n$ tends to infinity.
In this section, we state the lemmas and sketch the proofs. The full proofs can be found in \cref{sec:proof-RO}.

In order to find the optimal sequence $t^*$ we start by studying the success probability of algorithm $ALG_t$, for any sequence $t$, sample rate~$p$ and instance size~$n$. We prove that in fact the worst case for this class of algorithms is when $n$ is very large. The approach of approximating the problem when $n$ is large by a continuous time problem was pioneered by Bruss~\cite{B84} and has been used for different optimal stopping problems (see e.g. \cite{CCJ15, IKM06}).

\begin{restatable}{lemma}{lemMonotonicityOfALGt}
  For any sequence $t$ and sampling probability $p$, the success probability of $ALG_t$ in ROS$p$ decreases with $n$.
  \label{lem:MonotonicityOfALGt}
\end{restatable}

To prove the lemma the idea is to inductively couple the realizations of the arrival times in instances of sizes $n$ and $n+1$. We show that if $ALG_t$ fails for a given realization of the arrival times of the largest $n$ values in the instance of size $n$, then $ALG_t$ also fails for any possible realization of the arrival time of the smallest (the $n+1$-th largest) value, in the instance of size $n+1$. This implies that the probability of failure increases with $n$.

By \cref{lem:MonotonicityOfALGt} the success guarantee of $ALG_t$ is simply the limit of its success probability when $n$ grows to infinity. We calculate these probabilities and obtain an explicit formula for the limit in the following lemma. The formula turns out to be surprisingly simple.

\begin{restatable}{lemma}{lemSuccGuaranteeALGt}
  Fix a sequence $t$ and a sampling probability $p$. The success guarantee of $ALG_t$ in $ROSp$ is given by
  \begin{equation}
    \label{eq:SuccGuaranteeALGt}
    \sum_{i=1}^{\infty}  p^{i-1}\Bigg( 
    1-\max\{p,t_i\}
    - \int_{\max\{p,t_i\}}^1 \sum_{j=1}^i \frac{t-\max\{p,t_i\}}{t^j} \,dt
    \Bigg)\,. 
   \end{equation}
  \label{lem:SuccGuaranteeALGt}
\end{restatable}

We then focus our attention on optimizing this success guarantee. Surprisingly, it turns out that the problem of maximizing \cref{eq:SuccGuaranteeALGt} is separable and concave, so we can simply impose the first-order conditions to obtain the optimum. Perhaps even more surprising is that these first-order conditions are independent of $p$, and therefore, the optimal sequence $t^*$ is also independent of $p$, as the following lemma shows.

\begin{restatable}{lemma}{lemOptimalSequencet}
  Fix a sampling probability $p$. The sequence $t^*$ defined as the unique solution of the equations
  \begin{align}
    \ln\left( \frac{1}{t^*_i} \right) + \sum_{j=1}^{i-1}\frac{(1/t_i^{*})^j-1}{j}=1,\;\; \text{ for all } i\in\mathbb{N}\, ,
    \label{eq:OptimalSequencet}
  \end{align}
  maximizes \cref{eq:SuccGuaranteeALGt}. In particular, $t^*$ does not depend on $p$. 
  \label{lem:OptimalSequencet}
\end{restatable}

Now that we have the best algorithm in the family, we prove that its success guarantee is actually the best possible. To do this, we first characterize the algorithm that achieves the highest success probability for fixed sampling probability $p$ and instance size $n$.

For a non-decreasing function $\ell:[n]\rightarrow[n]$, we define the \emph{sequential-$\ell$-max algorithm} the following way. The algorithm accepts the $i$-th observed value (considering the values from $S$ and the ones that have been revealed from $V$) if it is the largest seen so far from $V$ and it is larger that the $\ell(i)$-th largest value from $S$. We prove that the optimal algorithm is in this class.

\begin{restatable}{lemma}{lemBestOrdinalAlgo}
  Fix a sampling probability $p$ and an instance size $n$. There is a function $\ell$ such that the sequential-$\ell$-max algorithm obtains the best possible success probability for instances of size $n$ of ROS$p$.
  \label{lem:BestOrdinalAlgo}
\end{restatable}

To conclude the optimality of $ALG_{t^*}$ we show that the success probability of the best sequential-$\ell$-max algorithm for each $n$ converges to \cref{eq:SuccGuaranteeALGt} for some sequence $t$, when $n$ grows to infinity. The idea behind this is to calculate the success probability of the algorithm, show that there is a limit for the optimal $\ell$ in a continuous space, and use a Riemann sum analysis to obtain \cref{eq:SuccGuaranteeALGt} in the limit.

\begin{restatable}{lemma}{lemOptimalityOfALGt}
  Fix a sampling probability $p$. For each $n\in \mathbb{N}$, choose $\ell_{p,n}$ so that the sequential-$\ell_{p,n}$-max algorithm achieves the best possible success probability for fixed $p$ and $n$. There exists a sequence $t$ such that the success probability of the sequential-$\ell_{p,n}$-max algorithm converges to \cref{eq:SuccGuaranteeALGt} when $n$ grows to infinity.
  \label{lem:OptimalityOfALGt}
\end{restatable}

Finally, we study the success guarantee of $ALG_{t^*}$ in the border values of $p$, and show that it actually becomes equal to the best possible among all algorithms. It is easy to see that the success guarantee is $1/e$ when $p=0$. Note that when $p=0$, \cref{eq:SuccGuaranteeALGt} simplifies to $t_1 \ln(1/t_1)$, and that from \cref{eq:OptimalSequencet} we obtain that $t^*_1=1/e$. Replacing gives the success guarantee of~$1/e$. The case when $p$ tends to $1$ is a bit more involved and requires some tedious calculations. We evaluate \cref{eq:SuccGuaranteeALGt} with the first order approximation $t^*_i\approx 1-\frac{c}{i}$, for some constant $c$. Then, we do a Riemann sum analysis to show that, when $p$ tends to $1$, this approximation converges to the explicit expression of Samuels~\cite{S82,S91} for $\gamma$.


\section{Knowledge of the parameters}\label{sec:Parameters-overview}

In this section we briefly discuss the impact of the knowledge of the parameters on the guarantees that can be obtained. There are two parameters for both AOS$p$ and ROS$p$: the number of elements $n$ and the sampling probability $p$. The performance of an algorithm can vary a lot depending on its presumed knowledge about these parameters.

For AOS$p$ we already discussed that knowledge of~$n$ is irrelevant in worst case terms. To complete the picture, we turn our attention to the cases when $p$ is unknown. First, if $p$ is unknown but $n$ is known, we show that the ratio of the number of samples to the total number of elements gives a good estimate of $p$, and that using $k$-max with this estimate is basically optimal. More specifically, assume we are given $h$ samples drawn independently with probability $p$ from an initial set of $n$ values and the other $n-h$ values form the online set. The \emph{$k$-max algorithm for unknown $p$} sets the threshold to the $k$-th largest sample, where $k = \left\lfloor \frac{n}{n-h} \right\rfloor$, and accepts the first value of the online set exceeding this threshold. 

	%
	%

\begin{restatable}{theorem}{unknownpopt}
	\label{unknownpopt}
	For AOS$p$ with known $n$ and unknown~$p$, the variation of the $k$-max algorithm for unknown $p$ achieves the best possible success guarantee up to a factor $1-\varepsilon$ with high probability.
\end{restatable}

To prove the theorem, we find the probability (as a function of $n$, $h$ and $p$) that this variant of the $k$-max algorithm chooses the maximum value of the online set, and then use standard concentration arguments.  

Second, for AOS$p$ where both $p$ and $n$ are unknown, we show that no non-trivial guarantee can be obtained. The intuition behind this strong negative result results from the situation in which the algorithm is given very few samples. In this case, it does not know whether the instance is very short (in which case it should stop early), or the sampling probability is very low (in which case it should wait longer).

For ROS$p$, we have shown that the optimal algorithm $ALG_{t^*}$ does not depend on $p$, and knowledge of the uniform random arrivals suffices to obtain the optimal guarantee. Therefore, 
$ALG_{t^*}$ achieves the best possible success guarantee, even when $n$ is unknown. On the other hand, if $p$ is unknown and $n$ is known and large, then we can sample uniform random arrival times for each value and obtain with $ALG_{t^*}$ the best success guarantee. Indeed, the sampled arrival times themselves will provide a sharp estimate of $p$.

On a more applied note, whenever it is reasonable to assume that the values come in random order, it is usually also safe to assume that this random order comes from random arrival times. In case the arrival times are random but not uniform, the time thresholds $t^*$ can be transformed using the distribution function of the arrival times and again obtain the optimal success guarantee.


\section{Proof of \cref{thmAO}}
\label{sec:proof-AO}

This section provides the full details of the proof of the second statement of \cref{thmAO}, which is that no algorithm can achieve a better success guarantee than the $k$-max algorithm. The short proof of the first claim was presented in \cref{lem:k-max-guarantee}.
%
%
\thmAO*
As mentioned before, we prove the negative results for AOS$p$ by proving them for the last zero problem. These proofs heavily rely on the conflict graph. \cref{subsec:preliminaries-conflict} contains the preliminaries for the negative results. It formally introduces the conflict graph, proves some structural results and highlights its connection to both deterministic and randomized algorithms. 
Then, as a warm-up for the more complicated proof, we start with the negative result for deterministic algorithms for $p=1/2$ in \cref{subsec:p-half}. \cref{subsec:general-p} generalizes this to both randomized algorithms as well as general values of $p$. Finally, in \cref{subsec:n-known} we show that allowing the player to have even exact knowledge of the length $n$ of the sequence of bits, will be (in the worst-case) unnecessary. 

\subsection{Preliminaries: Last zero problem}\label{subsec:preliminaries-lastzero}



As outlined before, we prove the negative result by introducing the last zero problem, whose negative results directly imply the negative results for AOS$p$.

Let the \emph{norm} of a sequence of bits be the number of ones it has. The number of bits in such a sequence $s$ is called its length or size. 
The numbering of the entries of a sequence $s$ is counted starting from $1$.

\begin{Definition}
The \emph{last zero problem} with probability $p$ is the following.

\begin{enumerate}
\setlength\itemsep{0pt} 
\item An adversary picks a size $n$.
\item An instance is created the following way: a sequence of bits of length $n$ is generated, where in each position, independently, the number is 1 with probability $p$ and 0 otherwise.
\item The player is given the norm of the sequence.
\item Then the player sees the bits one after the other, and for each of them decides whether to stop or to continue. 
\item The player wins if she stops on the last 0 of the sequence.  
\end{enumerate}
\end{Definition}

Note that the fact that the player does not know the size $n$ is crucial, as otherwise the game is trivial. Thus, it does not make sense to analyze the algorithm for a given size; we have to prove that no algorithm can perform well on all sizes.

As proved in \cref{sec:OverviewAOS}, negative results for the last zero problem imply the same bounds for algorithms for AOS$p$.


\proplastzeroequivalent*

For the remainder of this section, we consider the last zero problem. An instance in the last zero problem can be described by a finite string of bits. We introduce the shorthand notation $0^\ell$ and $1^\ell$ for the string of length $\ell$ consisting of only zeros and ones respectively.


%
\subsection{Preliminaries: Conflict graph}\label{subsec:preliminaries-conflict}

We now formalize the intuition of \cref{subsec:overview-AO-negative} about the conflict graph. We first describe its generic structure, independent of the sampling probability $p$, without weights on the instances or any reference to success guarantees. Then we continue by describing how to measure the performance of a deterministic algorithm in this framework using probabilistic weights on the instances. 
\subsubsection{Conflict graph structure}


We first define what it means that two instances are in conflict. For an instance $I$ we denote by $I[a,b]$ the instance $I$ restricted to the positions $a$ to $b$ (both included). Consider two instances $I_1$ and $I_2$ of size $n_1$ and $n_2$ respectively with $n_1<n_2$, both containing at least one 0. Let $r$ be the position of the last 0 in $I_1$. The instances $I_1$ and $I_2$ are \emph{in conflict} if they have the same norm and $I_1[1,r]=I_2[1,r]$. The following lemma outlines why we care about this notion.

\begin{lemma}\label{lem:conflict-success}
No deterministic algorithm can win in two conflicting instances. 
\end{lemma}

\begin{proof}
Consider a deterministic algorithm that wins in $I_1$, meaning, this algorithm stops at position $r$. Note that at any position $j\leq r$, the knowledge of the algorithm up to that point consists of the norm of the instance and $I[1,j]$. Now run the same algorithm on $I_2$. Since the algorithm is deterministic and has the same information available at every point in time, it must make the exact same decision at every $j\leq r$. In particular, it stops at position $r$. However, since $I_2$ has the same norm as $I_1$ but a larger size, there must be a zero after position $r$ in $I_2$, and the algorithm loses in $I_2$. 
\end{proof}

We now define the conflict graph, which is the formal object described by \cref{fig:table1234}.  

\begin{Definition}
The \emph{conflict graph} is an infinite graph in which the nodes correspond to all finite strings of bits. There is an edge between nodes $s_1$ and $s_2$ if and only if the corresponding instances of the last zero problem are in conflict and the size of $s_2$ is one bit larger than the size of $s_1$.
\end{Definition}

As every node corresponds to a unique instance and vice versa, we will use these terms interchangeably.

When we draw the conflict graph, we order the nodes by increasing size as in \cref{fig:table1234}. We define a \emph{monotone path} as a (possibly infinite) path in the conflict graph where the nodes correspond to consecutive increasing sizes. For example, in \cref{fig:table1234}, (01, 001, 0010) is a monotone path.

\begin{lemma}\label{lem:monotone-path-conflict}
Two instances are in conflict if and only if they are linked by a monotone path in the conflict graph.
\end{lemma}
\begin{proof}
Let $I_1$ and $I_2$, be two instances of size $n_1$ and $n_2$ respectively that are in conflict, with $n_1<n_2$. By definition, they have the same norm, and have the same substring up to the last zero of $I_1$. Consider the following instance $I_3$: take $I_2$, and remove the last zero.

This instance (if it is not $I_1$) satisfies the two conditions above, thus is in conflict with $I_1$. It is also in conflict with $I_2$: they share the same prefix up to the last zero of $I_3$ and have the same norm. By repeating this operation (removing the last zero) until we get $I_1$, we get a series of instances (including $I_1$ and $I_2$), that are in conflict with one another, and can be ordered in increasing consecutive sizes. These instances form a monotone path in the conflict graph. 
The other direction of the proof follows similarly.
\end{proof}

This lemma and its proof have several consequences for the structure of the conflict graph. The following lemma is immediate.

\begin{lemma}\label{lem:conflict-relation}
Given an instance $I$ of size $n$, the instances of size $n+1$ in conflict with $I$ 
are the nodes that can be obtained by inserting a new zero anywhere after the last 0 of $I$. 
In the other direction, $I$ is in conflict with only one instance of size $n-1$: the one where the last zero has been removed.
\end{lemma}

This lemma implies that a node has only one edge on its left. We refer to this node as its \emph{parent}. We define the \emph{degree} of a node in the conflict graph as the number of neighbors it has on its right, which we refer to as its \emph{children}. Furthermore, for a given size $n$ each node corresponds to a different instance of zeros and ones, so we have $2^n$ nodes in total. The degrees adhere to the following structure. 

\begin{lemma}\label{lem:degrees}
Consider all $2^n$ nodes corresponding to instances of size $n$. For every $i \in \{1,\ldots,n-1\}$, there are $2^{n-i}$ nodes of degree $i$. Concretely, half of these nodes have degree one, a quarter of the nodes have degree two, and so on until one node has degree $n$.

Moreover, a node with degree $k$ has exactly one child of degree $i$ for every $i \in \{1,\ldots,k\}$. 
\end{lemma}

\begin{proof}
By \cref{lem:conflict-relation}, an instance of size $n+1$ is in conflict with an instance of size $n$ if we add a 0 anywhere after the last 0 of the instance of size $n$. Therefore, every instance that ends in a 0 has degree one, since the new 0 can only be inserted in one place. This is true for half of the nodes. Similarly, we see that every instance that ends in 01 has degree two and this is a quarter of the nodes. In general, every instance that has suffix $01^i$ has degree $i+1$.

A node with degree $k$ has a suffix $01^{k-1}$. To create a child, we need to add a zero anywhere after the last 0. If we insert this 0 at the very end, we create a node of degree one. If we insert this 0 before the last 1, we create a node of degree two. In general, if we insert the 0 after the $i$-th 1 from the end, we create a node of degree $i$.
\end{proof}

\subsubsection{Algorithms and weights in the conflict graph}\label{subsubsec:conflict-graph-weights}

We now turn to the connection between algorithms and the conflict graph. We start by linking the structure of the conflict graph to deterministic algorithms.

\begin{lemma}\label{lem:algorithm-path}
	A deterministic algorithm can win in at most one of the instances of any monotone path in the conflict graph.
\end{lemma}

\begin{proof}
	By \cref{lem:monotone-path-conflict}, any two instances that are in a monotone path are in conflict, and by \cref{lem:conflict-success} an algorithm can win in at most one instance of a pair of conflicting instances.
\end{proof}

One can think of an algorithm for the problem as a partition of the nodes of the conflict graph into the nodes for which it wins and the nodes for which it loses.
\cref{lem:algorithm-path} gives a constraint on the structure of such a partition. 
Note that not all partitions correspond to a finite algorithm, but this is not an issue as we look for impossibility results (we will abuse terminology and use the word ``algorithm'' nevertheless).


More precisely, we will consider such a partition in the following structured way. We start from some size $N_0$, and ask the algorithm which nodes of this size it \emph{selects}, that is, in which instances it wins. This implies that the algorithm will not be able to select some instances in the future, namely the instances in conflict with any node of this selection. We say that these nodes that cannot be selected later are \emph{removed}. Then we will move on to the next size, and ask the algorithm to select instances among those that have not been removed yet. We continue this in an iterative fashion. 

We now continue by adapting the conflict graph to reflect the quality of an algorithm. For this, we extend the conflict graph to weighted nodes. We first define this properly and show how to measure the quality of an algorithm in the conflict graph. For now, we restrict ourselves to deterministic algorithms, which select a specific node either always or never. Afterwards we show that the arguments extend to randomized algorithms as well, that are allowed to select nodes with some probability.  

We define the \emph{weight} of a node as the probability that the corresponding instance results from the sampling process where we sample each of the $n$ elements independently with probability $p$. In particular, if an instance has size $n$ and norm $m$, then the weight of the corresponding node in the conflict graph is $p^{m} (1-p)^{n-m}$. Note that for a fixed size $n$, the weights of the instances of size $n$ sum to 1.

With this definition of the weights, the \emph{performance} of a deterministic algorithm for a fixed size $n$ in terms of the weighted conflict graph is the sum of the weights of the instances in which it wins. 
Then, the \emph{worst case performance} of an algorithm is the infimum of the performance of the algorithm over all sizes $n$. Note that the worst case performance of an algorithm for the last zero problem implies a bound on the success guarantee of any algorithm for AOS$p$, which is exactly the negative result we aim to prove in this section.

Fix a size $n$ and let $V_{n,i}$ be the nodes of size $n$ with degree $i$. Define $w_i$ as the total weight of the nodes in $V_{n,i}$. Moreover, define $w_{ij}$ as the sum of the weights of the nodes of size $n+1$ and degree $j$ that are in conflict with any node in $V_{n,i}$. Note that $w_{ij}$ is only positive for $j \leq i$ because of \cref{lem:degrees}. The following lemma can be seen as the weighted version of this lemma. 

\begin{lemma}\label{lem:structure-large-p}
For any size $n$, $w_i = p^{i-1} (1-p)$ and $w_{ij}=(1-p)w_i=p^{i-1} (1-p)^2$ for all $1 \leq j \leq i \leq n$. 
\end{lemma}


\begin{proof}
From the proof of \cref{lem:degrees}, we see that the instances of degree one are exactly these which have a 0 in the end. Summing over their individual weights will give us $w_1 = 1-p$, which is the probability of having a 0 as the last bit of an instance. In general, a node of degree $i$ ends in a 0 followed by $i-1$ ones. Accordingly, 
the probability of having an instance that ends with this suffix is $w_i = p^{i-1} (1-p)$.

Now consider an instance $I_1$ of size $n$ with degree $i$. It starts with $n-i$ unrestricted bits and its suffix is $01^{i-1}$. Now consider an instance $I_2$ of size $n+1$ with degree $j$ that is in conflict with $I_1$. Because it is in conflict with $I_1$, the first $n-i$ unrestricted bits are the same as $I_1$, as well as the 0 in entry $n-i+1$. Then, we interrupt the suffix of $1^{i-1}$ with an additional 0 such that $I_2$ has a suffix of $1^{j-1}$ in order to have degree $j$. Therefore, $I_2$ has the following structure. It starts with $n-i$ unrestricted bits, followed by $01^{i-j}01^{j-1}$.

Now consider the set of all instances that have the form of instance $I_1$, with certain bits in its unrestricted prefix of length $n-i$. Because of the suffix $01^{i-1}$, the weight of these instances can be computed as $p^{i-1}(1-p)$. On the other hand, the weight of all instances that have the form of the instance $I_2$ can be computed as $(1-p)p^{i-j}(1-p)p^{j-1} = p^{i-1}(1-p)^2$.
%
%
\end{proof}

Finally, we describe how the conflict graph can reflect a randomized algorithm.  
The difference is that it labels each instance with a \emph{selection probability} $q$, while a deterministic algorithm labels each instance either with a one or a zero (we either always select it, or we never do). Concretely, this means the following. Suppose the algorithm is faced with the last 0 in this instance, but it is not aware of this of course. Then the algorithm stops with probability $q$ (and wins in this instance). It does not stop with probability $1-q$, meaning it loses in this instance (but might win in instances of larger size that are in conflict with this instance).

The following lemma is the non-binary version of \cref{lem:conflict-success}. For its statement, we define the \emph{descendants} of an instance $I$ as one would expect: The set of nodes in the conflict graph that are connected to $I$ through a monotone path and that have a larger size than $I$.

\begin{lemma}\label{lem:removed-fraction}
If for some instance $I$ there is a selection probability $q$, then the probability of winning at any descendant is at most $1-q$.
\end{lemma}
\begin{proof}
	This follows from similar arguments as \cref{lem:conflict-success}. 
\end{proof}
%
It is important to note that this \emph{removed fraction} adds up: if an instance has selection probability $q$, and one of its descendants has selection probability $r$, then for any descendant of the second instance its removed fraction is $q+r$ and its selection probability is at most $1-q-r$. In other words, when a randomized algorithm reaches this particular descendant, it can assign at most a selection probability of $1-q-r$ to it.

Similarly, we define the performance of a randomized algorithm as its quality for a given size $n$, i.e., the product of the weight of a node multiplied by its selection probability, summed over all  instances of size $n$. The worst case performance is then the infimum over $n$ of these performances. The worst case performance of an algorithm for the last zero problem provides a bound on the success guarantee of any algorithm for AOS$p$.


\subsection{Warm up: Proof of \cref{prop:p-half}}\label{subsec:p-half}

As a warm up that introduces the main ideas behind the general proof, this section proves the special case of Theorem~\ref{thmAO} for deterministic algorithms for the case where $n$ is unknown but larger than some constant, and $p = 1/2$. Note that for $p=1/2$, all nodes of size $n$ have the same weight, namely $1/2^n$. The total fraction of selected nodes is therefore equal to the total weight of the selected nodes.
\propphalf*

Note that the success guarantee of the $k$-max algorithm, proved in \cref{lem:k-max-guarantee}, can also be proved now using the alternative 
perspective of the conflict graph. The $k$-max algorithm roughly selects low degree nodes in every size $n$ of the conflict graph in order to remove as little weight as possible from instances of larger size. A careful analysis indeed gives the same success guarantee $kp^k(1-p)$.

To prove that the $k$-max algorithm has optimal success guarantee (for this special case), we will bound the worst case performance of any deterministic algorithm by considering a special class of algorithms.

\paragraph{Canonical algorithms.}

More precisely, we consider a deterministic algorithm that starts by selecting some nodes in the conflict graph for a certain size $N_0$. Consequently, all descendants of the selected nodes will be removed. The algorithm will then continue to the nodes of size $N_0 + 1$ and select a subset of the nodes of this size that have not been removed. Then it will continue to the next size and iterate this procedure.
%
We will show that if the algorithm consistently selects at least a $1/4+\varepsilon$ fraction of the nodes for each size, this process cannot run forever, reaching a contradiction.

Before we proceed to the proof, we make a crucial observation. Note that \cref{lem:degrees} implies that two nodes of the same degree have children with the same degree distribution, and the same holds for their further descendants. By construction, it follows that the subtrees to the right of any two nodes of the same degree are \emph{isomorphic}. With this important observation at hand, we can prove that it suffices to restrict our attention to algorithms of a canonical form, in order to reduce the large variety of possible algorithms.
%
%

\begin{lemma}\label{lem:same-degree-swapping}
Consider the last zero problem for $p=1/2$.
Let $I_1$ and $I_2$ be two instances of the same size that have not been removed, and consider an algorithm that selects $I_1$ but does not select $I_2$.
Then there exists another algorithm that selects $I_2$ instead of $I_1$ and achieves the exact same success guarantee. 
\end{lemma}
\begin{proof}
Consider the instances $I_1$ and $I_2$ and an algorithm $A$ that selects $I_1$ but not $I_2$. Since $A$ selects $I_1$, the nodes to its right are removed. On the other hand, as $I_2$ is not selected and is alive, it can be that $A$ selects some node in its subtree.
As observed above, the subtrees rooted at $I_1$ and $I_2$ are isomorphic. Since any node has at most one edge to the left, these trees are also disjoint. 

Now consider the algorithm $B$ that selects the same nodes as $A$ except for the following. It selects $I_2$ instead of $I_1$, deselects every node that $A$ selected in the subtree of $I_2$ and instead selects the corresponding (according to the isomorphism) nodes in the subtree of $I_1$.
By construction, for every given size $n$, 
the nodes that both algorithms select carry the same weight, so the success guarantees are equal. 
\end{proof}

We can reduce the algorithms of interest even further by introducing the following important notion.
We say that
an algorithm follows a \emph{small degrees first strategy} if for any size considered, among the nodes that are not yet removed, it selects the nodes with the smallest degrees.

Note that this strategy does not define a single algorithm: many nodes have the same degree. Indeed, the $k$-max algorithm is closely related to these small degrees first strategies -- we will elaborate on this in the paragraph of \cref{lem:kmax-fills}.

\begin{lemma}[Small degrees first strategy]\label{lem:small-degree-first}
Consider the last zero problem for $p=1/2$.
For every algorithm, there exists an algorithm using the small degrees first strategy that achieves the same performance for every $n$.
\end{lemma}  

\begin{proof}
Consider an algorithm that does not follow the small degrees first strategy. Then there exists a size $n$ where it selects an instance $I_1$ of degree $k_1$ and does not select an instance $I_2$ of degree $k_2<k_1$.

Consider the part of the subtree rooted at $I_1$ that consists of its $k_2$ children of smallest degree and their subtrees. By the structure given by \cref{lem:degrees}, this subtree is isomorphic to the subtree of $I_2$. Then the same swapping argument as in the proof of \cref{lem:same-degree-swapping} between the subtree of $I_1$ and the tree of $I_2$ exhibits another algorithm with the same success guarantee that does follow the small degrees first strategy. 
\end{proof}

From now on, we restrict ourselves to considering algorithms that follow the small degrees first strategy.

\paragraph{The cover ratio.}

In order to reach a contradiction and prove \cref{prop:p-half}, we define the \emph{cover ratio} $\rho$ for an algorithm and a certain size $n$. It is defined as the sum of the weights of the instances of size $n$ that the algorithm either selects or removes. The removal of an instance is due to selecting an instance of smaller size that is connected by a monotone path to this instance. Denoting the set of selected and removed instances of size $n$ by $S$ and $R$ respectively, and the weight of an instance $I$ by $w(I)$, we can write $\rho = \sum_{I: |I|=n, I \in S \cup R}w(I)$. Note that this sum over only $S$ instead of $S \cup R$ is the performance of the algorithm for size $n$.

Note that in the special case that $p=1/2$, all instances have equal weight and therefore $\rho = (r+s)/2^n$ is just the fraction of the total number of instances of size $n$ that are either selected or removed.


The proof sketch of \cref{prop:warm-up} in \cref{sec:OverviewAOS} showed the intuition behind the proof. Here we state the formal arguments. The idea behind the proof is to show that
selecting strictly more than $1/4$ of the instances for many successive sizes implies that the cover ratio $\rho$ increases in such a way that at some point it is impossible to select that many instances. This shows by contradiction that there is no deterministic algorithm that has a success guarantee of $1/4 + \varepsilon$.

\cref{lem:small-degree-first} implies that we can restrict ourselves to a unique strategy for the algorithm as follows. For a size $n$, select a $1/4+\varepsilon$ fraction of the non-removed instances in increasing order of degrees (with an arbitrary order for the instances of same degree). Then the algorithm repeats this for the non-removed instances in the next size $n+1$, which we refer to as the next \emph{step}. Without loss of generality, we can assume that we start at size $N_0$ with no removed nodes.

We now analyze the dynamics of the process, and in particular the dynamics of the cover ratio $\rho$. First, observe that at size $N_0$, no nodes have been removed so far. As the algorithm selects a $1/4+\varepsilon$ fraction of the nodes and half of all these nodes have a degree of 1, the algorithm  selects only nodes of degree 1. For a certain number of sizes, starting from $N_0$, the algorithm can select only degree 1 nodes. We call this the \emph{first phase} of the algorithm. 

\begin{claim}
	Consider the last zero problem for $p=1/2$ and an algorithm as described above.
	After $t$ steps in the first phase of the algorithm, the cover ratio $\rho$ is $(\frac{1}{4}+\varepsilon)\cdot\sum_{i=1}^{t}\frac{1}{2^{i-1}}$.
\end{claim}
\begin{proof}
	We prove the claim by induction. For the base case $n = N_0$ we have $\rho = 1/4 + \varepsilon$, which corresponds to the formula of the claim.
	Now suppose that the lemma holds for some size $n+t-1$, so $\rho=(\frac{1}{4}+\varepsilon)\cdot\sum_{i=1}^{t-1}\frac{1}{2^{i-1}}$. We first determine the fraction of removed nodes in the next size $n+t$. Since each node of degree 1 removes one node of the next size, the number of nodes removed for size $n+t$ is the same. However, as there are twice as many instances in total in size $n+t$, the fraction is half this number, namely $(\frac{1}{4}+\varepsilon)\cdot\sum_{i=1}^{t-1}\frac{1}{2^{i}}$. The fraction of selected nodes is $1/4+\varepsilon$, thus in total the cover ratio becomes
	\[
	\rho=
	\left(\frac{1}{4}+\varepsilon\right)
	\cdot
	\left(\sum_{i=1}^{t-1}\frac{1}{2^{i}}+1\right)
	=
	\left(\frac{1}{4}+\varepsilon\right)
	\cdot
	\left(\sum_{i=1}^{t}\frac{1}{2^{i-1}}\right) \,. \qedhere
	\]
\end{proof}

Note that the term $(\frac{1}{4}+\varepsilon)\cdot\sum_{i=1}^{k}\frac{1}{2^{i-1}}$ goes asymptotically to $\frac{1+\varepsilon}{2}$ as $k$ grows, for some $\varepsilon>0$. In particular, this means that at some point it exceeds the value of $1/2$, which is the total fraction of nodes with degree 1. This implies, in turn, that the algorithm is forced at some point to start selecting degree 2 nodes in addition to degree 1 nodes. This is the start of a \emph{second phase}, where the algorithm needs to select degree 2 nodes, in order to keep selecting a $1/4+\varepsilon$ fraction of the nodes for each size.

\begin{claim}
	Consider the last zero problem for $p=1/2$ and an algorithm as described above.
	In the second phase of the algorithm, the cover ratio $\rho$ grows by at least $\varepsilon$ at each step.
\end{claim}

\begin{proof}
	Let us consider a size $n$ where $\rho >1/2$, say $\rho = 1/2 + \delta$ for some $\delta>0$. Then for size $n+1$ the situation is the following. First, the $1/2$-fraction of nodes of size $n$ remove $1/4$ of the nodes of size $n+1$ (since all these nodes have degree 1). Then, by ~\cref{lem:degrees}, the $\delta$ fraction of degree 2 nodes remove one instance of degree 1 and one instance of degree 2 in the next size. That is, in size $n+1$, a $(1/4+\delta/2)$-fraction of the degree 1 nodes and a $\delta/2$-fraction of the degree 2 nodes are removed in total.
	
	The algorithm must now select a $(1/4+\varepsilon)$-fraction of the nodes that have not been removed. Following the small degrees first strategy, the algorithm chooses the remaining $1/4-\delta/2$ fraction of degree 1 nodes, and a $\delta/2+\varepsilon$ fraction of the degree 2 nodes. In total, for size $n+1$ we have $\rho = 1/2 + \delta + \varepsilon$, and the claim follows.
\end{proof}

These claims imply \cref{prop:p-half} as follows.

\begin{proof}[Proof of \cref{prop:p-half}]
	In the second phase of the algorithm, $\rho$ increases by $\varepsilon$ in each step. Therefore, at some point the cover ratio becomes too large to select only degree 1 and 2 nodes and the algorithm is forced to start selecting degree 3 nodes. Note that in this third phase $\rho$ also grows by at least $\varepsilon$ at each step, since selecting a node of degree 3 is even worse than selecting a node of degree 2: It removes the same number of degree 1 and 2 nodes, but in addition it removes degree 3 nodes.

	The same holds true for further phases of the algorithm in which it selects nodes of even higher degree. Due to this increase of at least $\varepsilon$ in each step, at some point $\rho$ becomes strictly larger than $3/4-\varepsilon$. Therefore, the algorithm cannot select an $1/4+\varepsilon$ fraction of the nodes any more. Therefore, no algorithm can achieve a success guarantee of $1/4 + \varepsilon$ for any $\varepsilon > 0$. 
\end{proof}
\subsection{Generalization to any value of $p$}\label{subsec:general-p}

In this section, we generalize the previous results beyond the case of deterministic algorithms for $p = 1/2$. Building on the intuition of the previous section, but using quite different techniques, we show what is the best possible success guarantee that \emph{any} algorithm can achieve. We then link our $k$-max algorithm to the conflict graph, such that we finally reach the main takeaway point of the section: The $k$-max algorithm, although very simple, is optimal for all values of $p$. We first focus on the family of deterministic algorithms and prove the optimality of $k$-max there. Then, we show how one can adapt the proof to include also randomized algorithms.

\paragraph{Local operators and average performance.}

The main reason the proof techniques of the previous section need to be adapted is the fact that instances of a given size do not have the same weight anymore, and therefore,
the swapping argument used in \cref{lem:same-degree-swapping} and \cref{lem:small-degree-first} is no longer true. 
Thus, we transform a strategy using moves that
select and deselect nodes from instances of different sizes: \emph{local operators}. 
%
These local operators might decrease the fraction of selected nodes in a specific size while increasing it for another size. To resolve this, we introduce the notion of the \emph{average performance} of an algorithm in the window $[n,n+t]$, which is simply the average of its performance on sizes $s\in[n,n+t]$. We will show that there exists a set of local operators that can be used to improve the average performance.

Informally, the argument is then as follows. The $k$-max algorithm is very consistent in the sense that it selects the same total weight for every size. This means that its average performance is approximately equal to the infimum of its performance for every size (i.e., its success guarantee). 
Therefore, if a strategy would outperform the $k$-max algorithm, it would also exceed the average performance in every window. In this section, we show that the latter is a contradiction.
%


To prove that certain local operators improve the average performance in the next lemma, we say an algorithm is \emph{valid} if it selects at most one node along each monotone path in the conflict graph. 


\begin{lemma}\label{lem:local-moves}
Consider a valid deterministic algorithm with a certain average performance in a window $[n,n+t]$. Applying the following local operators yields a new valid algorithm whose average performance in this window is at least as good as the former algorithm.
\begin{enumerate}
\item \label{item:local-moves-large-degree} If the algorithm selects a node of degree $d > 1/(1-p)$ for some size $s\in [n,n+t-1]$: Deselect it and select all its children.
\item \label{item:local-moves-small-degree} If the algorithm has not selected nor removed a node of degree $d \leq 1/(1-p) $: Select it and remove all its descendants (in particular, deselect its selected descendants).
\end{enumerate}
\end{lemma}

\begin{proof}
The fact that the resulting algorithm is valid again is clear. We prove that these local operators do not decrease the average performance.

Consider the first local operator and a node of degree $d$ and weight $w$ of size $s \in [n,n+t-1]$. After applying the operator, the performance of the algorithm in size $s$ is decreased by $w$. By \cref{lem:structure-large-p}, the total weight of its children is $dw(1-p)$, which is larger than $w$ for $d>1/(1-p)$.


Now consider the second local operator. Let $A_1$ be the algorithm before applying the second local operator and $A_2$ the resulting strategy afterwards. We will construct a reversed sequence of valid algorithms that starts at $A_2$, iteratively selects and deselects some nodes and ends in $A_1$, where in every step the average performance does not increase. This will prove the claim.

Consider a valid algorithm $A$ in this reversed sequence (the ``current'' algorithm) from which we will construct its predecessor algorithm $A'$. Let $v$ be the node that $A_1$ neither selects nor removes and consider the subtree $T$ rooted at $v$ for the remainder of this argument. Let $S$ be the set of nodes in $T$ that $A_1$ does not select, but that the current algorithm $A$ does select. Among the nodes in $S$, let $v'$ be an arbitrary node of minimum size. There are two cases to consider.

First, if $A_1$ does not select any of the descendants of $v'$, deselect $v'$ in the newly constructed algorithm $A'$. This clearly does not improve the average performance from $A$ to its predecessor $A'$.

Second, consider the other case where $A_1$ selects at least one of the descendants of $v'$. Denote the weight of $v'$ by $w'$. Then, to turn $A$ into $A'$, deselect $v'$ and selects all its descendants. Note that this replaces a node of degree $d' \leq d \leq 1/(1-p)$ and weight $w'$ by a set of at most $d'$ nodes of weight $w'(1-p)$, having total weight $d'w'(1-p) \leq w'$. So the average performance of $A'$ is at most the average performance of $A$. 

By starting at algorithm $A_2$ and iteratively applying these two cases, we create a sequence of valid algorithms that converge to the initial algorithm $A_1$. Since the average performance does not increase in this direction, this means that from $A_1$ to $A_2$ the average performance does not decrease and the proof is complete.
\end{proof}

\paragraph{Fill-in strategy.}
Using these local operators that improve the average performance, we can define the following.
%
%
The \emph{fill-in strategy} for a window $[n,n+t]$ scans the sizes in increasing order,
selects all the non-removed instances of degree up to $\left\lfloor \frac{1}{1-p} \right\rfloor$ for each size $s\in[n,n+t-1]$, and all the non-removed instances for size $n+t$. 


\begin{lemma}\label{lem:fill-in-optimal}
The fill-in strategy has optimal average performance for any window $[n,n+t]$.
\end{lemma}
\begin{proof}
Consider an optimal strategy that is not the fill-in strategy. There are three cases. In the first case, a non-removed node of size $s \in [n,n+t-1]$ of degree at most $\left\lfloor \frac{1}{1-p} \right\rfloor$ is not selected. But in this case, applying the second operator of \cref{lem:local-moves} improves the average performance, which is a contradiction. In the second case, a non-removed node of size $s \in [n,n+t-1]$ of degree strictly larger than $\left\lfloor \frac{1}{1-p} \right\rfloor$ is selected. Now we can apply the first operator of \cref{lem:local-moves} to improve the average performance, and we have a contradiction. In the last case, a non-removed node of size $n+t$ is not selected. But selecting it will also improve the average performance, which is again a contradiction.
\end{proof}

With the optimal fill-in strategy at hand, we now proceed to describe the $k$-max algorithm in the conflict graph and finally show that the worst case performance of the fill-in strategy does not exceed the success guarantee of the $k$-max algorithm to conclude the proof of the negative results of Theorem~\ref{thmAO}.


\paragraph{The $k$-max algorithm in the conflict graph.}
To link the fill-in strategy to the $k$-max algorithm, we need to analyze the dynamics of the $k$-max algorithm in the conflict graph.
As a starting point, we will describe which instances the algorithm selects for $p \in [1/2,2/3)$. Note that for such a value, $k = \left\lfloor \frac{1}{1-p} \right\rfloor = 2$, so the algorithm sets the second largest sampled value as a threshold (i.e., stops with the first 0 after the second-to-last 1 in the last zero problem). This implies that for any given size $n$, it obtains the last zero (i.e., the online element with the maximum value) in the instances which end in 110 or 101. Similarly, for $p \in [2/3,3/4)$, the algorithm successfully selects the last zero in instances that end in 1110, 1101 or 1011.

We analyze its dynamics in the conflict graph in the following lemma. We will need the concept of the \emph{$m$-cut suffix} of an instance, which is the last $m$ bits in case the instance has at least $m$ bits and the entire instance otherwise.

\begin{lemma}\label{lem:kmax-fills}
Consider the instances in the conflict graph of size $n \geq k+1$ and consider the $k$-max algorithm that starts at size $k+1$ and iteratively considers instances of increasing size. For every size, it selects the non-removed nodes that have norm at least $k$ as well as degree at most $k$.
\end{lemma}

\begin{proof}
	Consider the conflict graph for $n \geq k+1$ with selected and removed nodes by the 
	$k$-max algorithm and 
	suppose by contradiction that the lemma is false. Then either a node of norm less than $k$ is selected, or a node of degree more than $k$ is selected, or a node that has norm at least $k$ as well as degree at most $k$ is not selected.
	
	
	In the first case, there are less than $k$ samples, thus the algorithm sets a threshold of zero and accepts the first online value. So the algorithm only wins in this instance if the first online value is the maximum online value, i.e., the instance contains only one 0. But since there are less than $k$ samples, this instance has size at most $k$. Contradiction.
	
	In the second case, note that a node that has degree more than $k$ has a suffix consisting of one 0 followed by at least $k$ 1s. In such an instance, however, the 
	$k$-max algorithm loses, so it does not select such a node. Contradiction.
	
	In the third case, consider a node $v$ of norm at least $k$ and degree at most $k$ that is selected. Without loss of generality we assume that $v$ is the node with these properties of smallest size among all nodes with these properties. Let the degree of $v$ be $d \leq k$ such that its suffix is $01^{d-1}$. 
	Consider the $k$-cut suffix of $v$ and note that it contains at least one 0. Now, as long as the $k$-cut suffix of $v$ contains more than one 0, remove the last 0 of $v$. Consider the unique resulting instance $v'$ of this procedure whose $k$-suffix contains exactly one 0. Note that the size of $v'$ is at least $k+1$ as its norm is at least $k$. Since $v'$ has norm at least $k$ as well as degree at most $k$, and $v$ was the smallest (in terms of size) such node that was not selected, the $k$-max algorithm already selected $v'$. But then $v$, being a descendant of $v'$, was removed and therefore could not be selected in the first place, contradiction.
\end{proof}

Now that the behavior of the $k$-max algorithm on the conflict graph is clear, it is possible to analyze its success guarantee using the conflict graph. The possibility to analyze the success guarantee of an algorithm through the conflict graph is one of its key properties. Indeed, such an analysis yields the same success guarantee as the one claimed in \cref{lem:k-max-guarantee}. 

\paragraph{Connecting the fill-in strategy to the $k$-max algorithm.}
The previous lemma allows us to compare the performance of the fill-in strategy to the performance of the $k$-max algorithm. In fact, they select almost the same nodes in the conflict graph.

\begin{lemma}\label{lem:fill-in-k-max}
	Consider the fill-in strategy and the $k$-max algorithm for window $[n,n+t]$.
	
	If $n>1$, then for every size $s \neq n,n+t$, the fill-in strategy and the $k$-max algorithm select the same set of nodes. For sizes $s=n$ and $s=n+t$, the $k$-max algorithm selects a strict subset of the set of nodes selected by the fill-in strategy.
	
	 If $n=1$, they select the same set of nodes for size $s=n=1$ as well.
\end{lemma}

\begin{proof}
	%
	Suppose that we start with size $n>1$. This means that none of the instances of size $n$ have been removed. The fill-in strategy selects all nodes of degree up to $k = \lfloor 1/(1-p) \rfloor$. The $k$-max algorithm selects only such nodes that have norm at least $k$ as well, which is a strict subset.
	
	We will now prove that for sizes $n<s<n+t$, the set $S_1$ of nodes selected by the fill-in strategy is the same as the set $S_2$ of nodes selected by the $k$-max algorithm. It is clear that $S_2 \subseteq S_1$. We prove $S_1 \subseteq S_2$ by contradiction, so we assume there is a $v \in S_1 \setminus S_2$, i.e., $v$ has degree at most $k$ and norm less than $k$. Without loss of generality, we assume that $v$ has the smallest size among nodes in the set $S_1 \setminus S_2$. We consider two cases: $v$ has a parent $w$ of size $s-1$ or $v$ has no parent.
	
	In the first case, note that the degree $d$ of node $w$ is at most $k$. Otherwise, it would have suffix $01^d$ for $d \geq k$. But then its norm would be at least $k$ and the norm of its child $v$ would also be at least $k$, contradiction. So assume that the degree of $w$ is at most $k$. Then $w$ was selected by the fill-in strategy if it was not removed earlier. If $w$ was selected, $v$ was removed so could not be selected by the fill-in strategy, so $v \not\in S_1$, contradiction. If $w$ was not selected, that is because it was removed earlier. But it can only be removed earlier in case it is a descendant of a node that was selected by the fill-in strategy before. But in that case, $v$ was also removed, contradiction.
	
	In the second case, note that nodes without a parent are exactly the nodes that have at most one 0. In the single instance that contains no zeros, the k-max algorithm and the fill-in strategy make the same decision by definition, so we restrict ourselves to instances that contain exactly one 0. Since the norm of $v$ is less than $k$, the $k$-max strategy sets a threshold of 0 and wins, since the only 0 is the maximum 0. But then $v \in S_2$, contradiction. 
	
	We wrap up the first part of the proof by considering the size $s=n+t$. Here, the fill-in strategy selects all non-removed nodes, while the $k$-max algorithm selects all non-removed nodes that have degree at most $k$ and norm at least $k$. The set of removed nodes is the same and the set of non-removed nodes contains nodes of degree more than $k$ or norm less than $k$, so the fill-in strategy indeed selects more nodes.
	
	Finally, if $n=1$, both the fill-in strategy and the $k$-max algorithm select instance 0 and cannot win in instance 1, so in this case they select exactly the same nodes also in the first size of the window.
\end{proof}

Combining everything, we can now prove the negative result for deterministic algorithms. 


\begin{proof}[Proof of negative result of Theorem~\ref{thmAO} for deterministic algorithms]
	First, note that \cref{lem:fill-in-k-max} implies that the performance of the fill-in strategy and the $k$-max algorithm for the sizes $N_0$ and $N_0+t$ differs by at most 1 for each size, so the average performance of the fill-in strategy in $[N_0,N_0+t]$ is at most $2/(t+1) \leq 2/t$ more than the average performance of the $k$-max algorithm. As argued before, for some interval, the average performance of the $k$-max algorithm is arbitrarily close to $kp^k(1-p)$, since that is its worst case performance. Consider this interval. 
	
	To prove the theorem, suppose by contradiction that there exists an algorithm $A$ that achieves a performance of $kp^k(1-p) + \varepsilon$ for some $\varepsilon>0$ for every size $n$ (larger than some size $N_0$), where $k=\lfloor1/(1-p)\rfloor$. Consider a window $[n,n+t]$ (with $n \geq N_0$) for some $t>0$. Then the average performance of $A$ in $[n,n+t]$ is at least its worst case performance, which is $kp^k(1-p) + \varepsilon$. However, the average performance of the fill-in strategy in this window is (arbitrarily close to) $kp^k(1-p) + 2/t$ and this is optimal by \cref{lem:fill-in-optimal}. Therefore, for $t>2/\varepsilon$, this is a contradiction since $A$ cannot be better.
%
\end{proof}

Finally, we adapt the above proof to randomized algorithms by generalizing \cref{lem:local-moves} to the randomized setting. The rest of the proof follows immediately from the same arguments as for deterministic algorithms, so extending this lemma suffices to extend the negative results to randomized algorithms.

For a node $v$, let $q_s(v)$ and $q_r(v)$ be its selection probability and its removed fraction (cf. \cref{lem:removed-fraction}), respectively. Recall that a node selected with probability $q_s(v)$ removes a fraction $q_s(v)$ of its descendants. We call a randomized algorithm \emph{valid} if the sum of $q_s(v)$ over all vertices $v$ of a monotone path in the conflict graph is at most 1.
%
%

\begin{lemma}\label{lem:local-moves-randomized}
	Consider a valid randomized algorithm with a certain average performance in a window $[n,n+t]$. Applying the following local operators yields a new valid algorithm whose average performance in this window is at least as good as the former algorithm.
	
	\begin{enumerate}\setlength\itemsep{0pt} 
		\item \label{item:local-moves-large-degree-randomized} If the algorithm selects a node $v$ of degree $d \geq 1/(1-p)$ for some size $s \in [n,n+t-1]$ with probability $q_s(v) >0$, set $q_s(v) = 0$ and increase the success probability of its children by $q_s(v)$.
		\item \label{item:local-moves-small-degree-randomized} If for a node $v$ of degree $d \leq 1/(1-p)$ the algorithm sets $q_s(v) + q_r(v) < 1$, increase $q_s(v)$ by $\varepsilon = 1 - q_s(v) - q_r(v)$. Then for every descendant $v'$, set $q_s(v') = 0$ and $q_r(v') = 1$.
	\end{enumerate}
\end{lemma}

\begin{proof}
	For both local operators, the claim that applying them does not decrease the average performance follows from the arguments of \cref{lem:local-moves}, so in this proof we will show that both local operators result in a valid algorithm. Let $v$ be the node under consideration and for any node $w$ denote by $q'_s(w)$ and $q'_r(w)$ its selection probability and its removed fraction, respectively, after applying one of the local operators.
	
Consider the first local operator and any monotone path $P=(v,v_1,v_2,\ldots)$. Note that every monotone path contains exactly one child of $v$. Then
\begin{equation*}
\sum_{w \in P} q'_s(w)
= q'_s(v) + q'_s(v_1) + \sum_{i \geq 2} q'_s(v_i)
= 0 + \left(q_s(v_1) + q_s(v)\right) + \sum_{i \geq 2} q_s(v_i) = \sum_{w \in P} q_s(w) \, .
\end{equation*}
 So if the original algorithm was valid, so is the algorithm after applying this operator. 


For the second operator, note that we change $q_s(v)$ to $q_s(v)+\varepsilon = q_s(v) + 1 - q_s(v) - q_r(v) = 1 - q_r(v)$. Therefore, after applying the operator, we have $q_s(v) + q_r(v) = 1$. Since in general for any child $w$ of $v$ we have $q_r(w) = q_s(v) + q_r(v)$, we see that $q_s(w) \leq 1 - q_r(w) = 1 - 1 = 0$. The proof is complete.
\end{proof}
\subsection{Generalization for known $n$}\label{subsec:n-known}

We now prove that even exact knowledge of the size $n$ that the adversary picks for the instance does not help asymptotically. To do so, we first introduce a variant of the last zero problem.

\begin{Definition}
The \emph{colored last zero problem}  is the following;
 
 \begin{enumerate}
\setlength\itemsep{1pt} 
\item An adversary picks two integers $m$ and $n$, with $m\leq n$.
\item A sequence of bits of length $n$ is created where every entry independently has value 1 with probability $p$ and 0 otherwise.
\item We color the entries $1$ to $m$ with red, while the entries $m+1$ to $n$ are colored blue.
\item The player is given the size $n$, the number of red 1s and the number of blue 1s.
\item Then the player is presented with the bits one after the other, and for each of them decides whether to stop or to continue. 
\item The player wins if she stops on the last red 0 of the sequence.
\end{enumerate}
\end{Definition}

Note that now the player has three numbers to start with: the number of red samples $r$, the number of blue samples $b$ and the size $n$.

\begin{proposition}
The colored last zero problem is equivalent to a specific instance of AOS$p$ with known size $n$. Therefore, any negative result for the colored last zero problem also holds for AOS$p$. 
\end{proposition}

\begin{proof}
(Analogue of \cref{prop:last-zero-equivalent}.)
The player again only wins if she stops with the element of the online set with the largest value, only that now she knows in advance how many online elements she is going to observe. Imagine now that she is facing an instance of the following form: The first $m$ elements are assigned a series of positive strictly increasing values, and the remaining $n-m$ take arbitrary negative values. Thus, in this instance the player is aiming for the last non-sampled element among the first $m$. This is basically the same game as the colored last zero problem, where the red values correspond to the positive values and the blue values correspond to the negative ones.
\end{proof}

\begin{theorem}\label{k-max-opt-known-n}
In the colored last zero problem, no algorithm can achieve performance $kp^k(1-p)+\varepsilon$ on every size $n \geq N_0$ (for some $N_0>0$).
\end{theorem}

Intuitively, the colored last zero problem should not be much different from the case without colors: there is still an unknown point in the sequence where the player should stop, and there is still a sequence of bits before this point (the red bits). The only difference is that now $n$ is known and we are also given the total number of 1s in the last $n-m$ bits (the blue bits). At first sight these blue 1s seem useless, because the player wants to stop before reaching them. On the other hand, the fact that we know how many they are, gives an indication about the size of $n-m$ and this could be already enough to improve the performance. We show that this is not the case. To do so, we define a slightly different conflict graph, and study its structure to show that up to negligible terms the dynamics are the same as for the standard conflict graph.

\paragraph{Modified conflict graph}

For the colored last zero problem, $m$ basically plays the role that $n$ was playing before. Therefore, the different layers of the conflict graph correspond to the various sizes of $m$ in this case, and there is a separate conflict graph for each value of $n$. Note that the conflict graph has a finite number of layers as $m$ varies between $1$ and $n$.
  
A node of the graph is a couple $(S,b)$, where $S$ is a sequence of bits of length $m$, that represents the sequence of red bits, and $b$ is an integer that represents the number of blue 1s. The exact positions of 0s and 1s in the blue bits are irrelevant for our proof, only the total number of blue 1s matters.

Finally, just as before, the nodes have different weights, with the difference here that the weights also depend on $b$. In particular, the weight of a node $(S,b)$ is

\[
p^{r+b}(1-p)^{n-r-b}\binom{n-m}{b} \, .
\]

Indeed, the probability of having $r+b$ 1s in an instance of size $n$ when sampling with probability $p$ is $p^{r+b}(1-p)^{n-r-b}$, where $r$ is the number of red 1s. As we group together all the instances with $b$ blue 1s, we multiply by the total number of such instances.

\paragraph{Conflict structure}

Now let us consider the conflicts. One can see that two nodes $(S,b)$ and $(S',b')$ are in conflict if and only if $b=b'$, and $S$ is in conflict with $S'$ (in the sense of the standard conflict graph). Note that for an instance and its descendants the values $b$, $r$ and $n$ are the same. In other words, to move from size $m$ to size $m+1$ we can add a 0 in the appropriate position, just as in \cref{lem:conflict-relation}.

We now study the relation between the weights of an instance and its children.
Let $I_1$ be a node with a sequence $S$ of size $m$ and let $I_2$ be one of its children (note that $I_2$ has size $m+1$ and is in conflict with $I_1$). Let $p_1$ and $p_2$ be the weights associated with these nodes. We derive from the formula above that:

\[ 
\frac{p_2}{p_1} = \frac{\binom{n-m-1}{b}}{\binom{n-m}{b}}=\frac{n-m-b}{n-m}
\] 

Having defined the modified conflict graph, we are now ready to show the main result of this section.

\begin{proof}[Proof of Theorem~\ref{k-max-opt-known-n}]
Consider again the ratio $p_2/p_1$. The expected value of $b$ is of course $(n-m)p$, but this will not be the case for all instances that we consider. For large values of $n-m$ though, we can apply standard concentration arguments (see e.g.\ \cref{Hoeffding}) and obtain that with high probability we have

\begin{alignat*}{2}
\frac{(n-m)- (n-m)p - \varepsilon}{n-m} &\le &\frac{p2}{p1} &\le \frac{(n-m)- (n-m)p + \varepsilon}{n-m} \qquad \Longleftrightarrow\\
1 - p - \varepsilon' &\le{} &\frac{p2}{p1} &\le 1 - p + \varepsilon',
\end{alignat*}
where $\varepsilon' = \frac{\varepsilon}{n-m}$. From here it is easy to observe that when $\varepsilon$ takes a value very close to 0, so does $\varepsilon'$. Furthermore, as $n-m$ grows, $\varepsilon'$ vanishes. Thus the modified conflict graph has the same weight distribution as in \cref{lem:structure-large-p} with high probability.

Therefore, with high probability, the modified conflict graph is (almost) the same as the weighted conflict graph from \cref{subsec:preliminaries-conflict}. Thus, we can follow again the arguments in \cref{subsec:general-p}, since they all hold in this case too. We end up with the same impossibility results, which hold here as well both for deterministic and for randomized algorithms.
\end{proof}

\section{Proof of Theorem \ref{thmRO}}
\label{sec:proof-RO}

We first prove the lemmas of \cref{sec:RO}, which imply most of the statements of the theorem. We restate the lemmas here for better readability. We conclude by showing that the optimal success guarantee converges to $\gamma\approx 0.58$.
\lemMonotonicityOfALGt*
\begin{proof}
  Fix a sequence $t$ and a sampling probability $p$. We use a coupling argument between realizations of the arrival times in instances with $n$ and $n+1$ values. We start with an instance $\alpha_1,\dots,\alpha_{n+1}$, and assume the values are indexed in decreasing order. Consider a realization of the arrival times $\tau_1=\tau'_1,\dots,\tau_{n+1}=\tau'_{n+1}$ and couple it with the corresponding realization $\tau_1=\tau'_1,\dots,\tau_n=\tau'_n$ in the instance $\alpha_1,\dots,\alpha_n$. Assume that in the instance with $n$ values and for this particular realization of the arrival times, $ALG_t$ fails. This means that $V\setminus \{\alpha_{n+1}\}$ is non-empty and either $ALG_t$ never stops or it accepts a value that is not the maximum of $V\setminus\{\alpha_{n+1}\}$. Note that regardless of $\tau'_{n+1}$, the rankings of the values in $V\setminus \{\alpha_{n+1}\}$ are the same in both instances because $\alpha_{n+1}$ is smaller than all other values. Thus, if $\tau'_{n+1}<p$, $ALG_t$ does not succeed either when applied in the instance of $n+1$ values. On the other hand, if $\tau'_{n+1}>p$, we have to distinguish between two cases. If $ALG_t$ accepts $\alpha_{n+1}$, it fails, because $V\setminus\{\alpha_{n+1}\}$ is non-empty and then $\alpha_{n+1}$ cannot be the largest in $V$. If $ALG_t$ does not accept $\alpha_{n+1}$, then the behavior of $ALG_t$ in the rest of the variables is the same as in the instance with $n$ values and then it fails. 
  
  Since the distribution of $\tau_1,\dots,\tau_n$ is the same in both instances, we conclude with this argument that the probability that $ALG_t$ fails in the instance with $n+1$ values is at least as large as in the instance with $n$ values.
\end{proof}

\lemSuccGuaranteeALGt*
\begin{proof}
  We first calculate the success probability of $ALG_t$ for fixed $p$ and $n$ and then take the limit when $n$ tends to infinity.

  We say a value $\alpha_i$ is \emph{acceptable} for $ALG_t$ (for a particular realization of the arrival times) if $p<\tau_i$, for some $j\in \mathbb{N}$ we have that $t_j\leq \tau_i <t_{j+1}$, and $\alpha_i$ is larger than the $j$-th largest value in $S$. Now, note that if $\max V$ is not acceptable for $ALG_t$, then $ALG_t$ does not stop. This is because we restricted the sequence $t$ to be increasing, so values that arrive before $\max V$ are not acceptable, and values arriving after $\max V$ will not be the best seen so far from $V$. We use this to decompose the success probability as follows.
  \begin{align}
    \mathbb{P}(ALG_t \text{ succeeds})= \mathbb{P}(\max V \text{ is acceptable}) - \mathbb{P}(ALG_t \text{ stops before seeing } \max V)\,.
    \label{eq:SuccessProbDecomposition}
  \end{align}
  In this definition, if $V$ is empty we also say $\max V$ is acceptable. We first calculate the probability that $\max V$ is acceptable. Assume that the values are indexed in decreasing order, i.e., that $\alpha_1>\dots>\alpha_n$.
  \begin{align}
    \mathbb{P}(\max V \text{ is acceptable})&= \mathbb{P}(V=\emptyset) + \sum_{i=1}^n \mathbb{P}(\max V=\alpha_i)\cdot \mathbb{P}(t_i\leq \tau_i\,|\, \max V= \alpha_i)\notag\\
    &= p^n + \sum_{i=1}^n p^{i-1}(1-p)\cdot \frac{1-\max\{p,t_i\}}{1-p}\notag\\
    &= p^n + \sum_{i=1}^n p^{i-1}\left( 1-\max\left\{ p,t_i \right\} \right)\,.
    \label{eq:MaxVAcceptable}
  \end{align}

  By the same argument, $ALG_t$ stops before seeing $\max V$ if and only if at least one value arrives after $p$ and before the arrival time of $\max V$, and the maximum such value is acceptable. 
  \begin{align}
    &\mathbb{P}(ALG_t \text{ stops before seeing } \max V)\notag \\
    &= \sum_{j=1}^n \mathbb{P}(\max V=\alpha_j)\cdot \mathbb{P}(\text{maximum before }\max V \text{ is acceptable}|\max V=\alpha_j) \notag\\
    &= \sum_{j=1}^n \mathbb{P}(\max V=\alpha_i)
    \sum_{i=j}^{n-1} \mathbb{P}\Big(\text{max. in } [p,\tau_j) \text{ has rank } i \text{ and arrives in } [t_i, \tau_j) \Big|\max V=\alpha_j \Big) \notag \\
      &= \sum_{j=1}^n p^{j-1}(1-p)\sum_{i=j}^{n-1} \frac{1}{1-p}\int_{\max\{p,t_i\}}^1 
      \left( \frac{p}{t} \right)^{i-j}\cdot \frac{(t-\max\{p,t_i\})}{t} \notag \\
    &\hspace{50pt}\cdot \mathbb{P}(\text{at least } i \text{ values arrive before } t|\max V=\alpha_j, \tau_j=t) \, dt \notag \\
    &= \sum_{j=1}^n p^{j-1}\sum_{i=j}^{n-1} \int_{\max\{p,t_i\}}^1 
    \left( \frac{p}{t} \right)^{i-j}\cdot \frac{(t-\max\{p,t_i\})}{t}
    \Big( 1-B_{t,n-j}(i-j+1) \Big)\, dt \notag\\
    &= \sum_{i=1}^{n-1} p^{i-1} \int_{\max\{p,t_i\}}^1 \sum_{j=1}^i \frac{t-\max\{p,t_i\}}{t^j}
    \Big( 1-B_{t,n-j}(i-j+1) \Big)\, dt\, ,
    \label{eq:ProbALGStopsBefore}
  \end{align}
  where $B_{p,n}(x)=\sum_{i=0}^x \binom{n}{i} p^i(1-p)^{n-i}$ is the CDF of a Binomial distribution of parameters $p$ and $n$. Note that for any fixed integers $i$ and $j$, and time $t\in (0,1)$, $B_{t,n-j}(i-j+1)$ converges to $0$ when $n$ tends to infinty. Therefore, replacing \cref{eq:MaxVAcceptable} and \cref{eq:ProbALGStopsBefore} in the identity (\ref{eq:SuccessProbDecomposition}), and taking the limit when $n$ tends to infinity, we conclude the proof of the lemma.
\end{proof}

\lemOptimalSequencet*
\begin{proof}
  First, we relax the monotonicity constraint on the sequence of $t_i$'s. The resulting relaxed optimization problem is separable, i.e., optimizing over the entire sequence is equivalent to optimizing over each variable independently. For each $t_i$ we get the following equivalent problem.
  \begin{align*}
    \max_{t_i\in[0,1]} p^{i-1} \left(
    1-\max\{p,t_i\} - \int_{\max\{p,t_i\}}^1 \sum_{j=1}^i \frac{t-\max\{p,t_i\}}{t^j} \,dt
    \right)\,.
  \end{align*}
  Equivalently, we can remove the factor $p^{i-1}$ and restrict $t_i$ to be in $[p,1]$, obtaining
  \begin{align*}
    \max_{t_i\in[p,1]} 
    1- t_i - \int_{t_i}^1 \sum_{j=1}^i \frac{t-t_i}{t^j} \,dt\,.
  \end{align*}
  Denoting by $G_i(t_i)$ this objective function, we get that
  \begin{align*}
    \frac{d}{d t_i} G_i(t_i)= -1+ \int_{t_i}^{1} \sum_{j=1}^i \frac{1}{t^j} \, dt\,\text{,  and }\;\;
    \frac{d^2}{d t_i^2} G_i(t_i) = - \sum_{j=1}^i \frac{1}{t_i^j}\,.
  \end{align*}
  Therefore, $G_i(t_i)$ is a concave function and then the optimum is $\max\{p,t_i^*\}$, where $t_i^*$ is the solution of $\frac{d}{d t_i} G_i(t_i)=0$. In the original objective function $t_i$ appears always as $\max\{p,t_i\}$ so there we can simply take $t_i^*$ as the solution. Now we prove that $t_i^{*}$ is actually increasing in $i$, so it is also the optimal solution before doing the relaxation. In fact, $t_i^*$ satisfies
  \begin{align*}
    \int_{t_i^*}^1 \sum_{j=1}^i \frac{1}{t^j} = 1\,.
  \end{align*}
  Note that the left-hand side of this equation is decreasing in $t_i^*$, and is increasing in $i$. Thus, necesarily $t_i^*\leq t_{i+1}^*$, for all $i\geq 1$. We conclude that $t_i^*$ satisfies \cref{eq:OptimalSequencet} by simply integrating on the left-hand side of the last equation.
\end{proof}

\lemBestOrdinalAlgo*
\begin{proof}
  We study the optimal ordinal policy obtained with backward induction, and prove that it is in fact a sequential-$\ell$-max algorithm for certain $\ell$. Recall that we can assume the optimal policy is ordinal, so this algorithm will be optimal not only among ordinal algorithms.

  Denote by $X_i=\alpha_{\pi(i)}$ the $i$-th value, in the order of increasing arrival times. Denote by $R(X_1,\dots,X_j)$ the relative ranks of values $X_1,\dots,X_j$. In what comes, we use the notation $R(X_1,\dots,X_j)=x$ to condition on a particular realization $x$ of the relative ranks. Let $x$ be a realization of the ranks such that $X_j$ is the maximum in $V$ so far and has rank $r$. Then,
  \begin{align*}
    &\mathbb{P}\Big(X_j=\max V \Big| R(X_1,\dots,X_j)=x\Big)\\
    &= \mathbb{P}\Big(X_{j+1},\dots,X_n \text{ have overall rank at most $r+1$}\Big| R(X_1,\dots,X_j)=x \Big)\\
    &= \mathbb{P} \Big(X_{j+1},\dots,X_n \text{ have overall rank at most $r+1$}\Big)\\
    &= \prod_{s=0}^{r-1}\frac{j-s}{n-s}\,.
  \end{align*}
  The optimal policy is to accept $X_j$ if this probability is larger or equal than the probability of picking $\max V$ after rejecting $X_j$ if from $j+1$ onwards we use the optimal policy, conditional on $R(X_1,\dots,X_j)=x$.
  
  Let now $x'$ be a realization of $R(X_1,\dots, X_{j+1})$ such that the relative rank of the best of $V$ up to step $j+1$ is $r$. Suppose that conditional on $R(X_1,\dots,X_{j+1})=x'$, the probability of winning if we use the optimal strategy from $j+2$ onwards depends solely of $n$, $j+1$ and the relative rank $r$, for all possible ranks $r$. Denote this conditional probability by $W(n,j+1,r)$. We want to inductively prove that this is in fact true for all $n$, $j$ and $r$. It is of course true in the last step, when $j+1=n$, so we do induction on $j$. Let $x''$ be a realization of $R(X_1,\dots,X_j)$ such that the relative rank of the best of $V$ up to step $j$ is $r$. We have that
  \begin{align}
    &\mathbb{P}\Big(\text{win after } j \Big| R(X_1,\dots,X_j)=x''\Big) \notag\\
    &= \mathbb{P} \Big( X_{j+1} \text{ has relative rank }\geq r+1\Big| R(X_1,\dots,X_j)=x''\Big)\cdot W(n,j+1,r)\notag \\
    &\hspace{10pt} + \sum_{r'=1}^r \mathbb{P}\Big( X_{j+1} \text{ has relative rank } r' \Big| R(X_1,\dots,X_j)=x''\Big) \notag \\
    &\hspace{20pt} \cdot \max\left\{ W(n,j+1,r'), \prod_{s=0}^{r'-1}\frac{j+1-s}{n-s} \right\}\,. \label{eq:WinProbInduction}
  \end{align}
  But for all $x$,
  \begin{align*}
    \mathbb{P}\Big(X_{j+1} \text{ has relative rank } r'\Big| R(X_1,\dots,X_j)=x \Big) =\frac{1}{j+1}\,.
  \end{align*}
  This proves the inductive step. Therefore, $W(n,j,r)$ is well defined for all $n,j$ and $r$, and the optimal policy accepts $X_j$ that has relative rank $r$ and is the maximum so far in $V$ if and only if
  \begin{align*}
    \prod_{s=0}^{r-1} \frac{j-s}{n-s} \geq W(n,j,r)\,.
  \end{align*}
  From \cref{eq:WinProbInduction} it is easy to check that $W(n,j,r)$ is decreasing in $j$ for fixed $n,r$ and increasing in $r$ for fixed $n,j$.\footnote{At an intuitive level it is also easy to be convinced of this: as time passes it is harder to win, and if only low values (with large rank) have appeared, it is easier to win in the future.} Therefore the optimal policy is the sequential-$\ell$-max algorithm, for $\ell$ defined as
  \begin{align*}
    \ell(j)= \max\left\{ r: \prod_{s=0}^{r-1} \frac{j-s}{n-s} \geq W(n,j,r) \right\}\,.
  \end{align*}
\end{proof}

To prove \cref{lem:OptimalityOfALGt} we first find the success probability of the sequential-$\ell$-max algorithm for fixed $n$, $p$ and $\ell$.

\begin{lemma}
  Fix $n$, $p$ and a non-decreasing function $\ell$. Consider an integer $h$ such that $0\leq h<n$, and define $\hat{\ell}(i)=\min \left\{ \ell(i),h+1 \right\}$ for all $i\in [n]$. The success probability of the sequential-$\ell$-max algorithm, conditional on $|S|=h$, is given by
  \begin{align}
    &\frac{1}{n-h}\left( 1-\prod_{j=0}^{\hat{\ell}(h+1)-1} \frac{h-j}{n-j} \right)\notag \\
    &+\sum_{i=h+1}^{n-1}\left( 
    \sum_{r=h+1}^{i}\frac{1}{n-i} \left( \frac{1}{i-h}\prod_{j=0}^{\hat{\ell}(r)-1}\frac{h-j}{i-j} -\frac{1}{n-h}\prod_{j=0}^{\hat{\ell}(r)-1} \frac{h-j}{n-j} \right) -\frac{1}{n-h} \prod_{j=0}^{\hat{\ell}(i+1)-1}\frac{h-j}{n-j}
    \right)\,.
    \label{eq:SuccProbSeqEllMax}
  \end{align}
\end{lemma}
\begin{proof}
  We calculate first the probability of some events. For $i\in \{h+1,\dots, n\}$, denote by $A_i$ the event that the $i$-th element is the largest of V and the algorithm never stops. Notice that $A_i$ is equivalent to the event that the overall largest $\hat{\ell}(i)$ elements are in $S$, and the $i$-th element is the largest of $V$ (for this equivalence it is necessary that $\ell$ is non-decreasing). Therefore, we have that
\begin{align*}
  \mathbb{P}(A_i) = \frac{1}{n-h} \prod_{j=0}^{\hat{\ell(i)}-1} \frac{h-j}{n-j}\, .
\end{align*}
Note that this is $0$ if $\hat{\ell}(i)=h+1$.
Now, for $h+1\leq r \leq i\leq n$, define $B_{r,i}$ the event that the $r$-th element is the largest among positions $\left\{ h+1,\dots,i \right\}$ and the algorithm does not stop before $i+1$. This is equivalent to the event that the $r$-th element is the largest among positions $\left\{ h+1,\dots,i \right\}$ and the largest $\hat{\ell}(r)$ elements among positions $\{1,\dots,i\}$ are in $S$. Thus,
\begin{align*}
  \mathbb{P}(B_{r,i} ) =\frac{1}{i-h} \prod_{j=0}^{\hat{\ell}(r)-1} \frac{h-j}{i-j} \,.
\end{align*}
Now, note that $B_{r,i}\setminus A_r$ is the event that the $r$-th element is the largest among positions $\left\{ h+1,\dots,i \right\}$, but not of $V$, and the algorithm does not stop before $i+1$.  Note also that $A_r\subseteq B_{r,i}$. Therefore, the probability that the algorithm does not stop before $i+1$ and the maximum of $V$ is among positions $\{i+1,\dots,n\}$ is
\begin{align*}
  \sum_{r=h+1}^{i} \mathbb{P}(B_{r,i})-\mathbb{P}(A_r)= \sum_{r=h+1}^{i} \frac{1}{i-h} \prod_{j=0}^{\hat{\ell}(r)-1} \frac{h-j}{i-j} - \frac{1}{n-h} \prod_{j=0}^{\hat{\ell}(r)-1} \frac{h-j}{n-j} \,.
\end{align*}
Conditional on this event, the probability that the number in the $i+1$-th position is the largest of $V$ is $1/(n-i)$, because the relative order within positions $\{i+1,\dots,n\}$ is independent of this event. Thus, we obtained the probability that the $i+1$-th element is the largest of $V$ and the algorithm does not stop before $i+1$. To obtain the probability of winning in step $i+1$, we have to subtract the probability that the $i+1$-th element is the largest of $V$, but the algorithm never stops, i.e., $\mathbb{P}(A_{i+1})$. Therefore, the probability of winning at step $i+1$ is
\begin{align*}
  \frac{1}{n-i}\sum_{r=h+1}^{i} \left(\frac{1}{i-h} \prod_{j=0}^{\hat{\ell}(r)-1} \frac{h-j}{i-j} - \frac{1}{n-h} \prod_{j=0}^{\hat{\ell}(r)-1} \frac{h-j}{n-j}
  \right)
  - \frac{1}{n-h} \prod_{j=0}^{\hat{\ell}(i+1)-1} \frac{h-j}{n-j}\,.
\end{align*}
The probability of winning at step $h+1$ is slightly different, because the algorithm  never stops before it. In that case the probability of winning is
\begin{align*}
  \frac{1}{n-h}\left( 1-\prod_{j=0}^{\hat{\ell}(h+1)-1}\frac{h-j}{n-j} \right)\,.
\end{align*}
Adding these expressions concludes the proof of the lemma.
\end{proof}

\lemOptimalityOfALGt*
\begin{proof}
  First we show that the function $\ell$ that maximizes \cref{eq:SuccProbSeqEllMax}, in a certain sense converges to a function $\tilde{\ell}:(0,1)\rightarrow \mathbb{N}$. Then, we do a Riemann sum analysis to show that the success probability of the sequential-$\ell$-max algorithm converges to an expression in terms of $\tilde{\ell}$, and then we show that this can be equivalently expressed as \cref{eq:SuccGuaranteeALGt} for some sequence $t$.

  Except for terms that vanish when $n$ tends to infinity, \cref{eq:SuccProbSeqEllMax} can be rewritten as
  \begin{align}
    \sum_{r=h+1}^{n}\left( \sum_{i=r}^{n} \frac{1}{n-i}\left( \frac{1}{i-h} \prod_{j=0}^{\hat{\ell}(r)-1} \frac{h-j}{i-j}-\frac{1}{n-h}\prod_{j=0}^{\hat{\ell}(r)-1} \frac{h-j}{n-j}  \right) - \frac{1}{n-h} \prod_{j=0}^{\hat{\ell}(r)-1} \frac{h-j}{n-j}   \right)\,.  \label{eq:InProofSuccGuaranteeSeqEllMax}
  \end{align}
  To find the optimal $\ell(r)$ we simply maximize the following term as a function of $s$.
  \begin{align*}
  F_n(r,s)=  \sum_{i=r}^{n} \frac{1}{n-i}\left( \frac{1}{i-h} \prod_{j=0}^{s-1} \frac{h-j}{i-j}-\frac{1}{n-h}\prod_{j=0}^{s-1} \frac{h-j}{n-j}  \right) - \frac{1}{n-h} \prod_{j=0}^{s-1} \frac{h-j}{n-j} \,. 
\end{align*}
Between $s$ and $s+1$ the change is
\begin{align*}
  &F_n(r,s+1)-F_n(r,s)\\
  &= \sum_{i=r}^n \frac{1}{n-i} \left( \frac{\frac{h-s}{i-s}-1}{i-h}\prod_{j=0}^{s-1} \frac{h-j}{i-j} - \frac{\frac{h-s}{n-s}-1}{n-h}\prod_{j=0}^{s-1}\frac{h-j}{n-j} \right)-\frac{\frac{h-s}{n-s}-1}{n-h}\prod_{j=0}^{s-1}\frac{h-j}{n-j}\\
  &= \sum_{i=r}^n \frac{1}{n-i} \left( -\frac{1}{i-s}\prod_{j=0}^{s-1} \frac{h-j}{i-j} + \frac{1}{n-s}\prod_{j=0}^{s-1}\frac{h-j}{n-j} \right)+\frac{1}{n-s}\prod_{j=0}^{s-1}\frac{h-j}{n-j}\\
  &= \beta(n,s,h) \left( \sum_{i=r}^n \frac{1}{n-i}\left( 1 -\frac{n-s}{i-s}\prod_{j=0}^{s-1} \frac{n-j}{i-j} \right) +1 \right)\,,
\end{align*}
where $\beta(n,s,h)$ is a positive term, so the sign of this difference is not affected by it. The other term is decreasing in $s$, so $F_n(r,s)$ is maximized when this differences changes sign. In other words, it is maximized in
\begin{align*}
  \ell_n^*(i)= \min \left\{ s\in [n]:  \sum_{i=r}^n \frac{1}{n-i}\left( 1 -\prod_{j=0}^{s} \frac{n-j}{i-j} \right) +1 \leq 0\right\}\,.
\end{align*}
Now, doing a Riemann sum analysis, we have that $\tilde{\ell}(\tau)=\lim_{n\rightarrow \infty} \ell_n^*(\lfloor \tau n\rfloor)$ satisfies
\begin{align}
  \tilde{\ell}(\tau)= \min \left\{ s\in\mathbb{N}: \int_\tau^1 \frac{1}{1-t}\left( 1-\frac{1}{t^{s+1}}  \right)+1\leq 0 \right\}\,. \label{eq:InProofTildeEll}
\end{align}
Thus, interpreting \cref{eq:InProofSuccGuaranteeSeqEllMax} as a Riemann sum, and noting that $|S|/n$ converges to $p$ almost surely, we have that the success guarantee of the optimal policy converges to
\begin{align*}
  \int_{p}^1 \int_{\tau}^1 \frac{1}{1-t}\left( \frac{1}{t-p}\left( \frac{p}{t} \right)^{\tilde{\ell}(\tau)} - \frac{1}{1-p} p^{\tilde{\ell}(\tau)} \right)\, dt - \frac{1}{1-p} p^{\tilde{\ell}(\tau)}\, d\tau\,.
\end{align*}
From \cref{eq:InProofTildeEll} it is clear that $\tilde{\ell}$ is non-decreasing, so we can define the sequence $t_i= \inf \left\{ \tau\in [p,1]: \tilde{\ell}(\tau)\geq i \right\}$ and rewrite the limiting success guarantee in terms of it. Thus, we obtain
\begin{align*}
  \sum_{i=0}^\infty \left(\int_{t_i}^{t_{i+1}} \int_{\tau}^1 \frac{1}{1-t}\left( \frac{1}{t-p}\left( \frac{p}{t} \right)^{i} - \frac{1}{1-p} p^{i} \right)\, dt \, d\tau - \frac{t_{i+1}^{i}-t_i^{i}}{1-p}\right)\,.
\end{align*}
If we rearrange the terms, turning the integral from $t_i$ to $t_{i+1}$ into the difference between the integral from $t_i$ to $1$ and the integral from $t_{i+1}$ to $1$, we obtain
  \begin{align*}
    & \int_p^1\int_\tau^1 \frac{1}{(t-p)(1-p)}\,dt\,d\tau 
    - \frac{p}{1-p} \\
    &+ \sum_{i=1}^{\infty}\left(
    \int_{t_i}^1\int_\tau^1 \frac{1}{1-t}\left( 
    \frac{\left( \frac{p}{t} \right)^i- \left( \frac{p}{t} \right)^{i-1}}{t-p}
    -\frac{p^i-p^{i-1}}{1-p}
    \right)\,dt\,d\tau + \frac{t_i\left(p^i - p^{i-1}\right)}{1-p}
    \right)\\
    =& \frac{1}{1-p} - \sum_{i=1}^{\infty} p^{i-1}\left( 
    \int_{t_i}^1\int_\tau^1 \frac{1}{1-t} \left( 
    \frac{t-p}{t^i(t-p)} - \frac{1-p}{1-p}
    \right)\,dt\,d\tau + t_i\frac{1-p}{1-p}
    \right)\\
    =&{} \frac{1}{1-p}- \sum_{i=1}^{\infty} p^{i-1}\left( 
    \int_{t_i}^1\int_\tau^1 \frac{1}{t^i(1-t)} \left( 
    1-t^i
    \right)\,dt\,d\tau + t_i    \right)\\
    =& \frac{1}{1-p} - \sum_{i=1}^\infty p^{i-1}
    \left( \int_{t_i}^1\int_\tau^1 \sum_{j=0}^{i-1} \frac{t^j}{t^i} \,dt\,d\tau +t_i \right)\\
    =& \sum_{i=1}^\infty p^{i-1} \left( 
    1-t_i -\int_{t_i}^1\int_\tau^1 \sum_{j=1}^i \frac{1}{t^j}\,dt\,d\tau    
    \right)\\
    =& \sum_{i=1}^\infty p^{i-1} \left( 
    1-t_i -\int_{t_i}^1 \sum_{j=1}^i \frac{t-t_i}{t^j}\,dt    
    \right)    \,.
  \end{align*}
  This concludes the proof, since we defined the $t_i$'s in a way that they satisfy $t_i=\max\left\{ p,t_i \right\}$.
\end{proof}

We use the first-order approximation $t'_i=1-c/i$, which turns out to be good enough to show the convergence to the full information case. For large $i$ we want $t'_i$ to be close to $t_i^*$, so we impose that in the limit $t'_i$ satisfies \cref{eq:OptimalSequencet}. This means that
\begin{align*}
  1&= \lim_{i\rightarrow \infty} \ln\left( \frac{1}{1-c/i} \right) +
  \sum_{j=1}^{i-1} \frac{(1-c/i)^{-j}-1}{j} \\
  &= \int_0^1 \frac{e^{cx}-1}{x}\, dx\,.
\end{align*}
With this in hand we prove the following lemma.

\begin{lemma}
  Let $t'_i=1-c/i$, where $c$ is the solution of $\int_0^1\frac{e^{cx}-1}{x}dx=1$. When evaluated in $t'$, \cref{eq:SuccGuaranteeALGt} tends to
\begin{align}
  \gamma = e^{-c}+(e^{-c}-1-c)\int_1^{\infty} x^{-1}e^{-cx}\, dx \approx 0.5801\,,
  \label{eq:SamuelsFormula}
\end{align}
  when $p$ tends to $1$.
  \label{lem:ConvergenceToSamuels}
\end{lemma}

\begin{proof}
  We analyze separately the sum when $p=\max\{p,t'_i\}$ and when $t'_i=\{p,t'_i\}$. We call the first part $V_1$, which includes the terms up to $i=\lfloor \frac{c}{1-p} \rfloor$, and $V_2$ the rest.

\begin{align*}
  V_1&= \lim_{p\rightarrow 1} \sum_{i=1}^{\left\lfloor\frac{c}{1-p}\right\rfloor} p^{i-1}
  \left( 1-p- \int_p^1
  \sum_{j=1}^{i} \frac{t-p}{t^j} \, dt \right)\\
  &=\lim_{p\rightarrow 1} \sum_{i=1}^{\left\lfloor\frac{c}{1-p}\right\rfloor} p^{i-1}
  \left( 1-p- \int_p^1 \,dt + \int_p^1 \frac{dt}{t^i} - \int_p^1
  \sum_{j=1}^{i} \frac{1-p}{t^j} \, dt \right)\\
  &=\lim_{p\rightarrow 1} \sum_{i=1}^{\left\lfloor\frac{c}{1-p}\right\rfloor} p^{i-1}
  \left( \frac{p^{-(i-1)}-1}{i-1} - (1-p)\ln(1/p) - (1-p)
  \sum_{j=2}^{i} \frac{p^{-(j-1)}-1}{j-1} \right)\\
  &= \lim_{p\rightarrow 1} \sum_{i=1}^{\left\lfloor\frac{c}{1-p}\right\rfloor}\frac{1-p^{i-1}}{i-1} - \lim_{p\rightarrow 1} \sum_{i=1}^{\left\lfloor\frac{c}{1-p}\right\rfloor} (p^{i-1}-p^i)\sum_{j=2}^i \frac{e^{-(j-1)\ln p}-1}{j-1}\\
  &=  \lim_{p\rightarrow 1} \sum_{i=1}^{\left\lfloor\frac{c}{1-p}\right\rfloor}\frac{1-(p^{\frac{1}{1-p}})^{(i-1)(1-p)}}{(i-1)(1-p)}(1-p) - \lim_{p\rightarrow 1} \sum_{i=1}^{\left\lfloor\frac{c}{1-p}\right\rfloor} (p^{i-1}-p^i)\sum_{j=2}^i \frac{e^{-\frac{(j-1)}{i} i\ln p}-1}{(j-1)/i}\cdot \frac{1}{i} 
\end{align*}
Interpreting these two sums as Riemann sums, we obtain
\begin{align*}
  V_1&=\int_0^c \frac{1-e^{-x}}{x} \, dx - \int_{e^{-c}}^1 \int_0^1 \frac{e^{-x\ln y}-1}{x}\, dx\, dy\\
  &= \int_0^c \frac{1-e^{-x}}{x} \, dx - \int_{e^{-c}}^1 \int_0^1 \frac{e^{-x\ln y}-1}{-x\ln y} (-\ln y) \, dx\, dy\\
  &= \int_0^c \frac{1-e^{-x}}{x} \, dx - \int_{e^{-c}}^1 \int_0^{-\ln y} \frac{e^x-1}{x}\, dx\, dy\\
  &= \int_0^c \frac{1-e^{-x}}{x} \, dx - \int_0^c \int_{e^{-c}}^{e^{-x}} \frac{e^x-1}{x} \, dy\, dx\\
  &= \int_0^c \frac{1-e^{-x}- (e^{-x} - e^{-c})(e^x-1) }{x} \, dx\\
  &= e^{-c}\int_0^c \frac{e^x-1}{x} \, dx\\
  &= e^{-c}\int_0^1 \frac{e^{cx}-1}{x} \, dx\\
  &= e^{-c}\,,
\end{align*}
where the last step comes from the definition of $c$. On the other hand, we have that
\begin{align*}
  V_2&= \lim_{p\rightarrow 1}\sum_{i=\left\lfloor \frac{c}{1-p}\right\rfloor+1}^\infty p^{i-1}\left( \frac{c}{i}- \int_{1-\frac{c}{i}}^1 \sum_{j=1}^{i} \frac{t-1+c/i}{t^j} \, dt \right)\\
  &= \lim_{p\rightarrow 1}\sum_{i=\left\lfloor \frac{c}{1-p}\right\rfloor+1}^\infty p^{i-1} \left( \frac{c}{i} - \int_{1-c/i}^{1}\, dt +
  \int_{1-c/i}^1 \frac{1}{t^i}\, dt
  -\int_{1-c/i}^1 \sum_{j=1}^i \frac{c/i}{t^j}\, dt
  \right)\\
  &= \lim_{p\rightarrow 1}\sum_{i=\left\lfloor \frac{c}{1-p}\right\rfloor+1}^\infty p^{i-1}\left( \frac{(1-c/i)^{-(i-1)}-1}{i-1}
  + \frac{c}{i}\ln(1-c/i)
  - \sum_{j=2}^{i} c\frac{(1-c/i)^{-(j-1)}-1}{i(j-1)}
  \right)\\
  &= \lim_{p\rightarrow 1}\sum_{i=\left\lfloor \frac{c}{1-p}\right\rfloor+1}^\infty (p^{i-1}-p^i) \frac{(1-c/i)^{-(i-1)}-1}{\frac{1-p}{-\ln p}(i-1)(-\ln p)}\\
  &\;\;\;- \lim_{p\rightarrow 1}\sum_{i=\left\lfloor \frac{c}{1-p}\right\rfloor+1}^\infty \frac{p^{i-1}-p^i}{\frac{1-p}{-\ln p} i(-\ln p)} \sum_{j=2}^i \frac{c\left( (1-c/i)^{-i\frac{j-1}{i}} -1 \right) }{j/i} \cdot \frac{1}{i}\,,
\end{align*}
where in the last equality we omitted a term that vanishes when $p$ tends to $1$. We again interpret the sums as Riemann sums.
\begin{align*}
  V_2 &= \int_0^{e^{-c}} \frac{e^{c}-1}{\ln(1/x)}\, dx
  - c \int_0^{e^{-c}} \frac{1}{\ln(1/x)} \int_0^1 \frac{e^{cy}-1}{y}\, dy\, dx \\
  &= (e^c -1 -c) \int_0^{e^{-c}} \frac{1}{\ln(1/x)} \, dx\\
  &= (e^{-c}-1-c) \int_1^{\infty} x^{-1}e^{-cx} \,dx\,.
\end{align*}
In the second equality we used the definition of $c$ and in the third one we performed a change of variables. Summing $V_1$ and $V_2$ we get \cref{eq:SamuelsFormula}.
\end{proof}


\section{Proofs for results regarding parameter knowledge}
\label{sec:proof-knowledge}
This section provides the full proofs for the results for AOS$p$ as well as ROS$p$ regarding the different assumptions on the knowledge of the parameters. The basic ideas are given in \cref{sec:Parameters-overview}.

\subsection{AOS$p$ with known $n$ and unknown $p$}

In this section we consider the case where $n$ is known and the probability $p$ is unknown. Let us recall the theorem.


\begin{restatable}{theorem}{unknownpopt}
	\label{unknownpopt}
	For AOS$p$ with known $n$ and unknown $p$, the variation of the $k$-max algorithm for unknown $p$ achieves the best possible success guarantee up to a factor $1-\varepsilon$ with high probability.
\end{restatable}

In this scenario, we are given a set $S$ of $h$ samples, drawn independently from an initial set consisting of $n$ values in total, using some (unknown, but existing) value of $p$. The remaining $n-h$ samples form the online set $V$. 
We will show that adapting the $k$-max algorithm with the parameters that are known to the player achieves the best possible success guarantee.

\begin{Definition}[The $k^{\text{th}}$-max algorithm for unknown $p$]\label{def:k-max-unknown-p}
	Assume we are given $h$ samples drawn independently with probability $p$ from an initial set of $n$ values and the other $n-h$ values form the online set. The $k$-max algorithm sets the threshold to the $k$-th largest sample, where $k = \left\lfloor \frac{n}{n-h} \right\rfloor$, and accepts the first value of the online set that is above the threshold.
\end{Definition}

Intuitively, this algorithm boils down to the $k$-max algorithm that we described previously, where we estimate $p$ as $\hat{p} = h/n$ and use $\hat{p}$ to determine the desired value of $k$.
We will now prove its approximation guarantee and the fact that this is tight.

\begin{lemma}\label{k-max-unknown-p}
	For a given sample set $S$ with $h$ values and an online set $V$ with $n-h$ values, the $k$-max algorithm chooses the maximum value of the online set with probability 
	\begin{equation*}
	\Pr[\text{Win}] =\sum_{h=0}^{n} \left\lfloor \frac{n}{n-h} \right\rfloor \left( \frac{h}{n} \right)^{\left\lfloor \frac{n}{n-h} \right\rfloor} \frac{n-h}{n} {n\choose h} p^h (1-p)^{n-h} \, ,
	\end{equation*}
	where $p$ is the probability of independently sampling a value from the initial set. 
\end{lemma}

\begin{proof}
	Assume that the values of the adversarial input $\mathcal{A}$ are sorted in decreasing order $\alpha_1 > \alpha_2 > \cdots > \alpha_n$.
	Let us call $p_h$ the probability that the $k$-max algorithm succeeds in a particular instance with $h$ samples and $S_h$ the event where $|S| = h$. Then the total probability that the $k$-max algorithm succeeds equals
	\begin{align*}
	\Pr[\text{Win}] &= \sum_{h=0}^{n} \Pr[k\text{-max algorithm wins} \mid S_h] \cdot \Pr[S_h] \\
	&= \sum_{h=0}^{n} p_h {n\choose h} p^h (1-p)^{n-h} \, ,
	\end{align*}
	since each value of the initial set is sampled independently with probability $p$. It remains to determine $p_h$. Conditioned on the fact that we end up with $h$ samples, all the different labelings (as a sample or online value) of the initial $n$ values are equally likely to happen. 
	There are ${n\choose h}$ different labelings, and each $\alpha_i$ is labeled as a sample in a $h/n$-fraction of the possible labelings and as an online value in the rest.
	
	Observe that the algorithm succeeds only if \emph{exactly one} of the $\left\lfloor \frac{n}{n-h} \right\rfloor$ largest values of the adversarial input ends up in the online set and the $(\left\lfloor \frac{n}{n-h} \right\rfloor+1)$-th largest ends up in the sample set.
	To compute the number of such labelings, first consider those such that $\alpha_1,\alpha_2, \dots, \alpha_{\left\lfloor \frac{n}{n-h} \right\rfloor+1}$ are all labeled as samples except for exactly one. From those, we can exclude the labelings that mark $\alpha_{\left\lfloor \frac{n}{n-h} \right\rfloor+1}$ as an online value, since in this case $s_{\left\lfloor \frac{n}{n-h} \right\rfloor}$ is larger than all the online values. 
	Therefore, we obtain 
	\begin{align*}
	p_h &= \left( \left \lfloor \frac{n}{n-h} \right \rfloor +1 \right) \left( \frac{h}{n} \right)^{\left\lfloor \frac{n}{n-h} \right\rfloor}  \left( \frac{n-h}{n} \right) - \left( \frac{h}{n} \right)^{\left\lfloor \frac{n}{n-h} \right\rfloor}  \left( \frac{n-h}{n} \right) \\
	&= \left \lfloor \frac{n}{n-h} \right \rfloor \left( \frac{h}{n} \right)^{\left\lfloor \frac{n}{n-h} \right\rfloor}  \left( \frac{n-h}{n} \right) \, ,
	\end{align*}
	and the lemma follows.
\end{proof}

The theorem follows from the following well-known concentration bound. Essentially, we can prove that the estimate $\hat{p}$ is accurate with high probability.

\begin{lemma}[Hoeffding's inequality for i.i.d.\ Bernoulli random variables~\cite{Hoeffding}]\label{Hoeffding}
	
	Let $X_1, X_2, \dots, X_n$  be i.i.d.\ Bernoulli random variables with parameter $p$ 
	and let $\bar{X}=\left(\sum\nolimits_{i=1}^n X_i\right)/n$. 
	Then for any $\varepsilon>0$,
	\[ \Pr\left[\left|\bar{X}-pn\right| \geq \varepsilon\right] \leq 2e^{-2n\varepsilon^2} \, .\]
	Alternatively, by setting $\delta=2e^{-2n\varepsilon^2}$ we get that
	\[ \left| \bar{X}-pn \right| \leq \sqrt{\frac{1}{2n}\ln\frac{2}{\delta}} \qquad \text{with probability at least } 1-\delta \,. \]
\end{lemma}


\begin{proof}[Proof of Theorem \ref{unknownpopt}]
	Consider an instance of AOS$p$ for a fixed unknown value of $p$ where the player is faced with $h$ samples. The proof follows straightforwardly from the above concentration bound. For the purpose of analysis, let $\varepsilon_1$ and $\varepsilon_2$ be such that
	\begin{equation*}
	\varepsilon_1 \leq 1 - \frac{ \left(\frac{h}{n}\right)^{ \frac{n}{n-h} } } {p^{\frac{1}{1-p}}} \quad \text{and} \quad \varepsilon_2 \leq 2e^{-2n} \, .
	\end{equation*}
	The first value is chosen such that the following holds.
	\begin{align*}
	& \varepsilon_1 \leq 1 - \frac{ \left(\frac{h}{n}\right)^{ \frac{n}{n-h} } } {p^{\frac{1}{1-p}}} \\ 
	\Leftrightarrow \quad & \varepsilon_1 \leq 1 - \frac{ \frac{n}{n-h} \left( \frac{h}{n} \right)^{\frac{n}{n-h}} \frac{n-h}{n}} { \left( \frac{1}{1-p} - 1 \right) p^{ \frac{1}{1-p}-1 } (1-p) } \\ 
	\Leftrightarrow \quad & \varepsilon_1 \leq 1 - \frac{ \left\lfloor \frac{n}{n-h} \right\rfloor \left( \frac{h}{n} \right)^{\left\lfloor\frac{n}{n-h}\right\rfloor} \frac{n-h}{n}} { \left( \left\lfloor \frac{1}{1-p} \right\rfloor \right) p^{ \left\lfloor \frac{1}{1-p} \right\rfloor } (1-p) } \\ 
	\Leftrightarrow \quad & \left\lfloor \frac{n}{n-h} \right\rfloor \left( \frac{h}{n} \right)^{\left\lfloor\frac{n}{n-h}\right\rfloor} \frac{n-h}{n} \geq \left( \left\lfloor \frac{1}{1-p} \right\rfloor \right) p^{ \left\lfloor \frac{1}{1-p} \right\rfloor } (1-p) \cdot (1-\varepsilon_1) \, .
	\end{align*}
	The second value is chosen such that \cref{Hoeffding} yields $\Pr\left[\left|\bar{X}-pn\right| < 1\right] \geq 1 - 2e^{-2n} \geq 1 - \varepsilon_2$. Therefore, with probability at least $1-\varepsilon_2$, we have
	\begin{equation*}
	\sum_{h=pn-\varepsilon}^{pn+\varepsilon} {n\choose h} p^h (1-p)^{n-h} = \left. {n\choose h} p^h (1-p)^{n-h} \right|_{h=pn} \geq 1 - \varepsilon_2 \, .
	\end{equation*}
	
	With these values at hand we can bound the success guarantee of \cref{k-max-unknown-p} as follows: With probability at least $1 - \varepsilon_2$ we get that
	\begin{align*}
	\Pr[\text{Win}] &= \sum_{h=0}^{n} \left\lfloor \frac{n}{n-h} \right\rfloor \left( \frac{h}{n} \right)^{\left\lfloor \frac{n}{n-h} \right\rfloor} \frac{n-h}{n} {n\choose h} p^h (1-p)^{n-h} \\
	& \geq \left\lfloor \frac{1}{1-p} \right\rfloor p^{\left\lfloor \frac{1}{1-p} \right\rfloor}(1-p) \cdot(1-\varepsilon_1) \cdot \sum_{h=pn-\varepsilon}^{pn+\varepsilon} {n\choose h} p^h (1-p)^{n-h} \\
	& \geq \left\lfloor \frac{1}{1-p} \right\rfloor p^{\left\lfloor \frac{1}{1-p} \right\rfloor}(1-p) \cdot (1 - \varepsilon_1) \cdot (1 - \varepsilon_2) \, .
	\end{align*}
	For any given $\varepsilon>0$, one can take $\varepsilon_1$ and $\varepsilon_2$ that adhere to the bounds above and such that $(1-\varepsilon_1)(1-\varepsilon_2) \leq (1-\varepsilon)$. This yields a success guarantee that is at least $1-\varepsilon$ times the success guarantee of the $k$-max algorithm for known $p$.
\end{proof}


\subsection{AOS$p$ with unknown $n$ and $p$}\label{subsec:both-unknown}

This section proves that in the adversarial order case where both $n$ and $p$ are unknown, the player cannot obtain a positive success guarantee.

\begin{theorem}\label{unknownpunknownn}
When both $p$ and $n$ are unknown, no algorithm can get positive success guarantee.
\end{theorem}

\begin{proof}
	Let $\varepsilon>0$. We will prove that it is not possible to achieve a success guarantee of $\varepsilon$.
	
	Consider the following new game for any $\delta > 0$. The adversary selects a size $n$ and generates an instance of this size with increasing values. Then, the adversary again selects $p$ appropriately, so that the probability that there is at least one sample is at most $\delta$ and the probability that there are no samples is at least $1-\delta$. Then the sampling process happens and the player faces the sequence. If at least one value is sampled, the player automatically wins, otherwise, she wins if and only if she selects the last non-sampled value.
	
	Consider the case where there are no sampled values. Since the player does not learn anything along the game, any deterministic algorithm waits $t-1$ values before it selects the $t$-th value. A randomized algorithm can be thought of as a distribution over the stopping times $t$. Since the domain of $t$ are all positive integers, it is not possible that this distribution has weight at least $\lambda$ for every size, for any constant $\lambda>0$. Therefore, on instances with stopping probability less than $\lambda$, the player only wins with probability at most $\lambda$. Such an instance occurs trivially with probability at most 1.
	
	Overall, in this new game, the player wins in at most $\delta + \lambda$ values. Taking e.g.\ $\delta$  and $\lambda$ slightly smaller than $\varepsilon/2$, the success guarantee of this game is less than $\varepsilon$.
	
	The proof for AOS$p$ with unknown $p$ and $n$ follows easily now. The adversary chooses values of $n$ and $p$ as above. In case there are no sampled values, both games are the same, since in both cases the player has the same information and the same available strategies. In case there is at least one sampled value, the player wins in the new game with probability 1 and in AOS$p$ with probability strictly less than 1. Therefore, the success guarantee of AOS$p$ is at most the success guarantee of the new game, which is less than $\varepsilon$.
\end{proof}

\paragraph{Acknowledgements.}
The authors would like to thank two anonymous reviewers for their helpful comments that contributed to a better exposition of the paper, as well as pointing out the related work in \cite{CCJ15,KM20,BGSZ20}.

Jose Correa and Laurent Feuilloley were partially funded by ANID grant CMM-AFB 170001 and by an Amazon Research Award. 
Andr\'es Cristi is supported by ANID under grant PFCHA/Doctorado Nacional/2018-21180347. Alexandros Tsigonias-Dimitriadis is supported by the Alexander von Humboldt Foundation with funds from the German Federal Ministry of
Education and Research (BMBF) and by the German Research Foundation (DFG) within the Research Training Group AdONE (GRK 2201).
Part of the work was done when Tim Oosterwijk was visiting the Universidad de Chile, supported by ANID under grant FONDECYT 1181180.


\begin{thebibliography}{99}
	
	\bibitem{AI15} Allaart, P., Islas, J. A sharp lower bound for choosing the maximum of an independent sequence. {\em Journal of Applied Probability}, 53:1041--1051, 2015. 
	
	\bibitem{AKW14}
	Azar, P., Kleinberg, R., Weinberg., S.M. Prophet inequalities with limited information. SODA 2014.
	
	\bibitem{BK19}
	Beyhaghi, H., Kleinberg, R. Pandora's problem with nonobligatory inspection. EC 2019.
	
	\bibitem{BGSZ20} Bradac, D., Gupta, A., Singla, S., Zuzic, G. Robust algorithms for the secretary problem. ITCS 2020.
	
	\bibitem{B84}
	Bruss, F.T. A Unified Approach to a Class of Best Choice Problems with an Unknown Number of Options. {\em Annals of Probability}, 12(3):882--889, 1984.
	
	\bibitem{B00}
	Bruss, F.T. Sum the odds to one and stop.
	{\em Annals of Probability}, 28(3):1384--1391, 2000.
	
	\bibitem{CCJ15} Chan, H.T.H., Chen, J., Jiang, S.H.C. 
	Revealing Optimal Thresholds for Generalized Secretary Problem via Continuous LP: Impacts on Online K-Item Auction and Bipartite K-Matching with Random Arrival Order. SODA 2015.
	
	\bibitem{CCES20} Correa, J., Cristi, A., Epstein, B., Soto, J. The two-sided game of googol and sample-based prophet inequalities. SODA 2020.
	
	\bibitem{CDFK19} Correa, J., Dutting, P., Fischer, F., Schewior, K. Prophet Inequalities for IID Random Variables from an Unknown Distribution. EC 2019. 
	
	\bibitem{D18}
	Doval, L. Whether or not to open Pandora’s box. {\em Journal of Economic Theory}, 175:127--158, 2018.
	
	\bibitem{D63} Dynkin, E.B. 
	The optimum choice of the instant for stopping a Markov process. {\em Soviet Math. Dokl.} 4:627--629, 1963.
	
	\bibitem{EHLM20} Esfandiari, H., HajiAghayi, M., Lucier, B., Mitzenmacher, M. Prophets, secretaries, and maximizing the probability of choosing the best, AISTATS 2020.
	
	\bibitem{F89} Ferguson, T.S. Who solved the secretary problem? {\em Statistical Science}, 4(3):282--296, 1989.
	
	\bibitem{GM66} Gilbert, J., Mosteller, F. Recognizing the maximum of a sequence. {\em Journal of the American Statistical Association}, 61:35--73, 1966.
	
	\bibitem{Hoeffding} Hoeffding, W. Probability inequalities for sums of bounded random variables. {\em Journal of the American Statistical Association}, 58(301):0 13--30, 1963.
	
	\bibitem{IKM06} Immorlica, N., Kleinberg, R.D., Mahdian, M. Secretary problems with competing employers. WINE 2006.
	
	\bibitem{KNR20} Kaplan, H., Naori, D., Raz, D. Competitive analysis with a sample and the secretary problem. SODA 2020.
	
	\bibitem{KM20} Kesselheim, T., Molinaro, M. Knapsack secretary with bursty adversary. ICALP 2020.
	
	\bibitem{KS77} Krengel, U., Sucheston, L. Semiamarts and finite values. {\em Bulletin of the American Mathematical Society}, 83:745--747, 1977.
	
	\bibitem{KS78} Krengel, U., Sucheston, L. On semiamarts, amarts, and processes with finite value.
	{\em Advances in Probability} 4:197--266, 1978.
	
	\bibitem{L61} Lindley D.V. Dynamic programming and decision theory. 
	{\em Journal of the Royal Statistical Society (Series C Applied Statistics)}, 10:39--51, 1961.
	
	\bibitem{MSM85}
	Moran, S., Snir, M., Manber, U. Applications of ramsey’s theorem to decision tree complexity. {\em Journal of the ACM}, 32(4):938–949, 1985.
	
	\bibitem{RWW20} Rubinstein, A., Wang, J.Z., Weinberg, S.M. Optimal single-choice prophet inequalities from samples. ITCS 2020.
	
	\bibitem{W79} Weitzman, M. Optimal search for the best alternative. {\em Econometrica}, 47(3):641--654, 1979.
	
	\bibitem{S82} Samuels, S. Exact solutions for the full information best choice problem. {\em Purdue Univ. Stat. Dept. Mimeo Series}, 82-17, 1982.
	
	\bibitem{S91} Samuels, S. Secretary Problems. In {\em Handbook of Sequential Analysis}, Chapter 16, (B. Ghosh, P. Sen, Eds.), CRC Press, 1991.

\end{thebibliography}
\end{document}